\pdfoutput=1
\documentclass{sig-alternate}

\usepackage[utf8]{inputenc}
\usepackage[english]{babel}
\usepackage{amsmath, amsthm, amssymb}
\usepackage{color}
\usepackage{booktabs}

\usepackage[hyphens]{url}
\usepackage{verbatim} 
\usepackage{multirow}
\usepackage{xspace}
\usepackage{graphicx}
\usepackage{epsfig}
\usepackage{times}
\usepackage{dsfont}
\usepackage{paralist}
\usepackage{balance}

\usepackage{algorithm}
\usepackage{algorithmicx}
\usepackage[noend]{algpseudocode}

\newcommand\todo{\textcolor{red}}

\usepackage{etoolbox}
\newcommand{\zerodisplayskips}{%
  \setlength{\abovedisplayskip}{4pt} 
  \setlength{\belowdisplayskip}{4pt}
  \setlength{\abovedisplayshortskip}{4pt}
  \setlength{\belowdisplayshortskip}{4pt}}
\appto{\normalsize}{\zerodisplayskips}
\appto{\small}{\zerodisplayskips}
\appto{\footnotesize}{\zerodisplayskips}

\let\oldenumerate\enumerate
\renewcommand{\enumerate}{
  \vspace{-0.7\topsep} 
  \oldenumerate
  \setlength{\itemsep}{1pt}
  \setlength{\parskip}{0pt}
  \setlength{\parsep}{0pt}
  \setlength{\topsep}{0pt}
  \setlength{\partopsep}{0pt}
}
\let\olditemize\itemize
\renewcommand{\itemize}{
  \vspace{-0.7\topsep}
  \olditemize
  \setlength{\itemsep}{1pt}
  \setlength{\parskip}{0pt}
  \setlength{\parsep}{0pt}
  \setlength{\topsep}{0pt}
  \setlength{\partopsep}{0pt}
}

\newcommand*{\bfrac}[2]{\genfrac{}{}{0pt}{}{#1}{#2}}

\DeclareMathOperator*{\argmin}{arg\,min}
\DeclareMathOperator*{\argmax}{arg\,max}
\DeclareMathOperator*{\average}{mean}

\newcommand{\relevance}{\textsc{Rel}}
\newcommand{\entityrelevance}{\textsc{ERel}}
\newcommand{\daterelevance}{\textsc{DRel}}

\newcommand{\temporalconstraints}{\textsc{Constraints}}
\newcommand{\EtoE}{\textsc{E2E}}
\newcommand{\EtoEPath}{\textsc{E2EPath}}
\newcommand{\EtoECooc}{\textsc{E2ECooc}}
\newcommand{\EtoD}{\textsc{E2D}}
\newcommand{\EtoDPath}{\textsc{E2DPath}}
\newcommand{\EtoDCooc}{\textsc{E2DCooc}}
\newcommand{\GtoE}{\textsc{G2E}}
\newcommand{\globalimportance}{\textsc{GlobalImportance}}

\newcommand{\sub}{s}
\newcommand{\pathtore}[1]{\mbox{$\pi_{re}(#1)$}}
\newcommand{\pathtots}[1]{\mbox{$\pi_{t}(#1)$}}

\newcommand{\reFun}{\textsc{RE}}
\newcommand{\subFun}{\textsc{Sub}}
\newcommand{\tFun}{\text{$\tau$}}

\newtheorem{theorem}{Theorem}[section]

\newtheorem{lemma}[theorem]{Lemma}
\theoremstyle{remark}

\theoremstyle{definition}
\newtheorem{definition}{Definition}[section]
\newenvironment{proof-idea}{\noindent{\it Proof Sketch.}\hspace*{0.1em}}{\qed\bigskip\\}

\makeatletter
\newtheorem*{rep@theorem}{\rep@title}
\newcommand{\newreptheorem}[2]{%
\newenvironment{rep#1}[1]{%
 \def\rep@title{#2 \ref{##1}}%
 \begin{rep@theorem}}%
 {\end{rep@theorem}}}
\makeatother
\newreptheorem{theorem}{Theorem}

\usepackage{mathtools}

\DeclarePairedDelimiter\floor{\lfloor}{\rfloor}

\usepackage{ stmaryrd }

\usepackage{dsfont}

\usepackage{booktabs} 

\newcommand{\sectionrule}{\addlinespace[0.5ex]}

\newcommand{\eat}[1]{}

\newcommand{\etc}{{\em etc.}}
\newcommand{\cf}{{\em cf.}}

\newcommand{\hide}[1]{}
\newcommand{\xhdr}[1]{\vspace{1mm}\noindent{{\bf #1.}}}

\newcommand{\eg}{\emph{e.g.}}
\newcommand{\ie}{\emph{i.e.}}

\newcommand{\systemname}{{\sc TimeMachine}}

\graphicspath{{./FIG/}}

\newfont{\mycrnotice}{ptmr8t at 7pt}
\newfont{\myconfname}{ptmri8t at 7pt}
\let\crnotice\mycrnotice%
\let\confname\myconfname%
\permission{Permission to make digital or hard copies of part or all of this work for personal or classroom use is granted without fee provided that copies are not made or distributed for profit or commercial advantage and that copies bear this notice and the full citation on the first page. Copyrights for third-party components of this work must be honored. For all other uses, contact the Owner/Author(s). Copyright is held by the owner/author(s).}
\conferenceinfo{KDD'15,}{August 10-13, 2015, Sydney, NSW, Australia.} 
\copyrightetc{ACM \the\acmcopyr}
\crdata{978-1-4503-3664-2/15/08. \\
DOI: http://dx.doi.org/10.1145/2783258.2783325 }
\begin{document}

\title{TimeMachine:\\ Timeline Generation for Knowledge-Base Entities}

\author{
\alignauthor
Tim Althoff*,
Xin Luna Dong\textsuperscript{\textdagger},
Kevin Murphy\textsuperscript{\textdagger},
Safa Alai\textsuperscript{\textdagger},
Van Dang\textsuperscript{\textdagger},
Wei Zhang\textsuperscript{\textdagger}\\
\vspace{2mm}
\affaddr{*Computer Science Department, Stanford University, Stanford, CA 94305}\\
\affaddr{\textsuperscript{\textdagger}Google, 1600 Amphitheatre Parkway, Mountain View, CA 94043}\\
\vspace{1mm}
\email{*althoff@cs.stanford.edu \hspace{1mm} \textsuperscript{\textdagger}\{lunadong, kpmurphy, safa, vandang, weizh\}@google.com} 
}



\maketitle

\begin{abstract}

We present a method called 
\systemname\
to generate a timeline of events and relations for entities in a knowledge base. 
For example for an actor, such a timeline should show the most important
professional and personal milestones and relationships such as works, awards, collaborations, and family relationships.
We develop three orthogonal timeline quality criteria that an ideal timeline should satisfy:
(1) it shows events that are {\em relevant} to the entity;
(2) it shows events that are {\em temporally diverse}, 
so they distribute along the time axis, avoiding visual crowding
and allowing for easy user interaction,
such as zooming in and out;
and
(3) it shows events that are {\em content diverse},
so they contain many different types of events ({\em e.g.}, for an actor,
it should show movies and marriages and awards, not just movies).
We present an algorithm to generate such timelines for a given time period and screen size,
based on submodular optimization and web-co-occurrence statistics
with provable performance guarantees.
A series of user studies using Mechanical Turk shows 
that all three quality criteria are crucial to produce quality timelines 
and that our algorithm significantly
outperforms various baseline and state-of-the-art methods.

\end{abstract}



\vspace{1mm}
\noindent {\bf Categories and Subject Descriptors:} H.2.8 {\bf
[Database Management]}: Database applications---{\it Data mining}

\noindent {\bf General Terms:} Algorithms, Experimentation. 


\vspace{-0.5\baselineskip}
\section{Introduction}
\label{sec:introduction}


\begin{figure*}[t]
	\centering
  \includegraphics[width=0.8\linewidth]{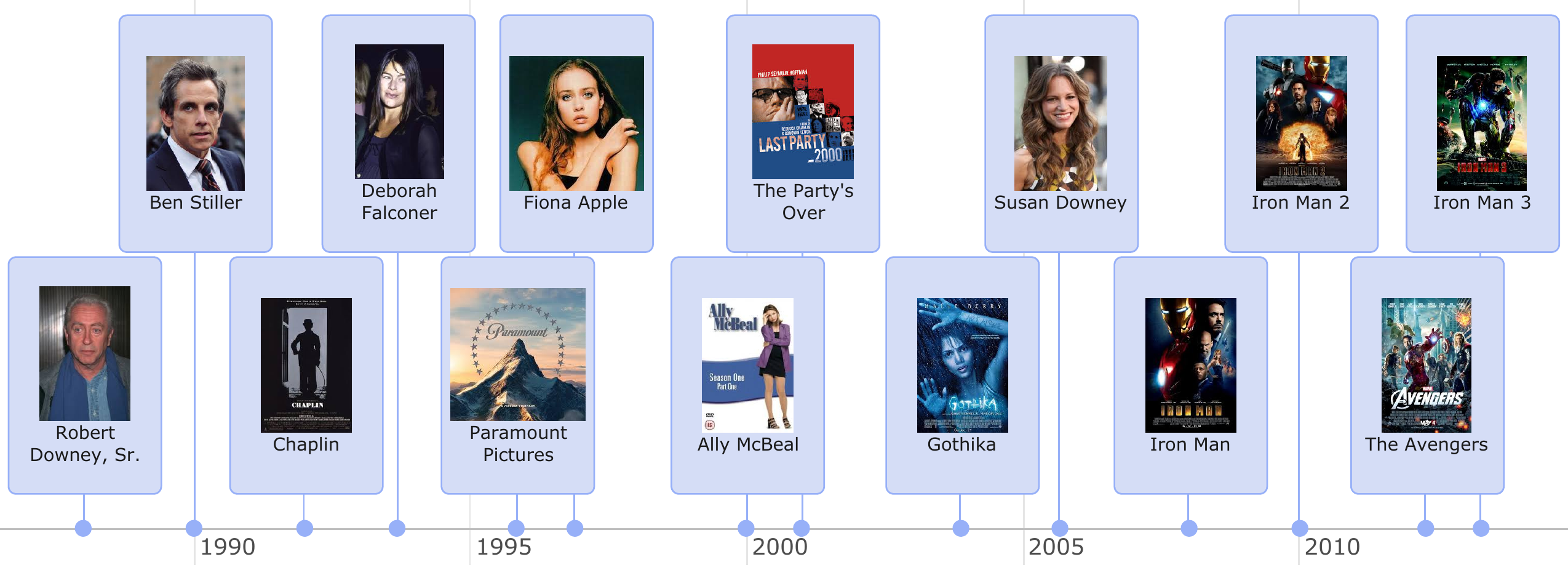}
 \vspace{-4mm}
	\caption{An example timeline for Robert Downey Jr. (American actor) generated by our proposed approach.
	 Note that the timeline is interactive and displays
         explanations for each event on hover
(see Figure~\ref{fig:timeline_rdjr_zoom}).
Furthermore, it can be dynamically zoomed.
}
	\label{fig:timeline_rdjr_full}
\end{figure*}




As the web and other technological advancements continue to bring down barriers for creation and distribution of information, 
relevant information is often buried in an avalanche of data, and locating it has become increasingly difficult~\cite{shahaf2012metro}.
Search engines have attempted to address this challenge~\cite{modern-ir}, 
but the volume and diversity of results can still be overwhelming,
even for simple entity queries~\cite{shahaf2013information}.
In many such cases, for instance when searching for a person or organization, 
an overview
of the most important events in an organized and readable format 
would serve users better,
ideally with interactive
features to enable further exploration.
A {\em timeline} with clickable key events arranged 
along a horizontal time axis would serve this need~\cite{tuan2011cate}.

Automatically generating timelines is very challenging.
To be specific, consider creating a timeline for the American actor Robert Downey Jr.
There are hundreds of possible candidate events and it is infeasible to display all of them.
Robert Downey Jr. is best known for his starring roles in the movies {\em Iron Man} and {\em Avengers}, 
but even for a single movie there are dozens of related events to display (production, release dates, opening, and award ceremonies).
In fact, one should not only focus on movies but provide a more holistic overview of his life and career.
This could include showing various family relationships 
(\eg, father Robert Downey Sr., ex-wife Deborah Falconer, or wife Susan Downey), 
important acting roles for his career (the movie {\em Chaplin} and TV
show {\em Ally McBeal}), 
and other notable works and professional relationships.
However, note that events might be related as well --- if one includes a movie award one might not want to display its release date separately but rather show a more diverse event instead.
Lastly, the timeline should be interactive to enable further exploration. 

Knowledge bases (KB) of timestamped facts such as 
Freebase~\cite{Bollacker08} 
or YAGO~\cite{suchanek2007yago}
have been used as the source of event information 
(in this paper we use Freebase).
Previous work has introduced timeline generation from KBs
through visualizing entity-level co-occurrence in news corpora~\cite{mazeika2011entity}, 
displaying events associated with an entity in YAGO~\cite{wang2010timely}, 
and generating context-aware timelines from Wikipedia~\cite{tuan2011cate}.
However, these works did not address the problem of selecting a subset of events but instead displayed all events~\cite{mazeika2011entity,wang2010timely},
or have used a static global ranking that does not capture dependencies between events and is therefore unable to encourage diversity~\cite{tuan2011cate}.
Furthermore, this existing work has not considered challenges raised by enabling user interaction
nor provided an empirical evaluation of the quality of the generated timelines.

\xhdr{Present work} In this paper, we develop an approach called 
\systemname\
to generate a timeline for a given entity of interest.
We develop three orthogonal timeline quality criteria:

\begin{enumerate}
  \item Relevance: Display only the most ``interesting'' or
          ``relevant'' events in an entity's    history.

  \item Temporal Diversity: Distribute events evenly along the
          temporal axis, to avoid visual crowding, and to allow easy
          interaction with the depicted events.

  \item Content Diversity: Display a diverse set of event types
(\eg, for an actor, do not only list the movies they have been in).

\end{enumerate}
Consequently, we propose a principled solution to timeline generation according to these criteria based on submodular optimization,
for which we both provide theoretical performance guarantees
and show empirical evidence of significant improvement over baseline
and state-of-the-art methods.

In Figure~\ref{fig:timeline_rdjr_full}, we show that our approach successfully
generates a timeline of relevant events that is diverse both in terms of content and time.
This timeline is also interactive in three ways.
First, when the user hovers over an
image,
we show various details, such as 
 ``Robert Downey Jr. starred in {\em The Avengers}, released on
 April 11, 2012''.
Second, the user can specify a particular
    time period, and a new timeline for that period will replace the current one.
    For example, Figure~\ref{fig:timeline_rdjr_zoom} shows the timeline
    for Robert Downey Jr. from 2007 until 2014 
    and gives an example event description.
The zoomed-in version focuses attention on more recent events,
such as his award for {\em Tropic Thunder}
and his role in the {\em Sherlock Holmes} movies.
Finally, the user can click on an entity icon, such as Susan Downey,
    and a timeline for this entity will be displayed.

Our approach involves the following two main steps, which are sketched in Figure~\ref{fig:system_overview}.
First, given a subject entity of interest,
we generate as many candidate events as possible
by searching for neighboring entities with timestamps in the given knowledge base (Section~\ref{sec:generation}).
We generate candidate events for all possibly interesting subjects offline.
%
Second, given a set of candidates and a time period of interest, we
select (online)
a diverse subset of the most relevant events subject to temporal diversity or layout constraints (Section~\ref{sec:selection}).
To do this, we maximize a submodular objective using various
relevance signals based on web co-occurrence, 
subject to these layout constraints.
We prove that our greedy algorithm for optimization
yields close-to-optimal solutions.
In addition, 
our algorithm allows for fast dynamic updates of the timeline based on
user interaction (zooming in or out).

We evaluate our proposed algorithm through a series of user studies with 1154 raters and compare it to
various simpler baselines and state-of-the-art approaches (Section~\ref{sec:evaluation}).
Our experiments show that users always significantly prefer our proposed method (60-91\% of timeline comparisons).
Further, we demonstrate that enforcing temporal diversity and content diversity significantly improves the results.

In summary, our main contributions are as follows:
\begin{enumerate}
  \item A design of a timeline search engine that efficiently 
        supports various types of user interaction.

  \item An algorithm for generating entity
          timelines based on submodular optimization
          and web-co-occurrence scores.

  \item An extensive user study of the relative importance of
          different signals for determining entity relevance
          and different notions of diversity.
\end{enumerate}
\vspace{-1.0\baselineskip}

\begin{figure}[t]
  \centering
  \includegraphics[width=.9\linewidth]{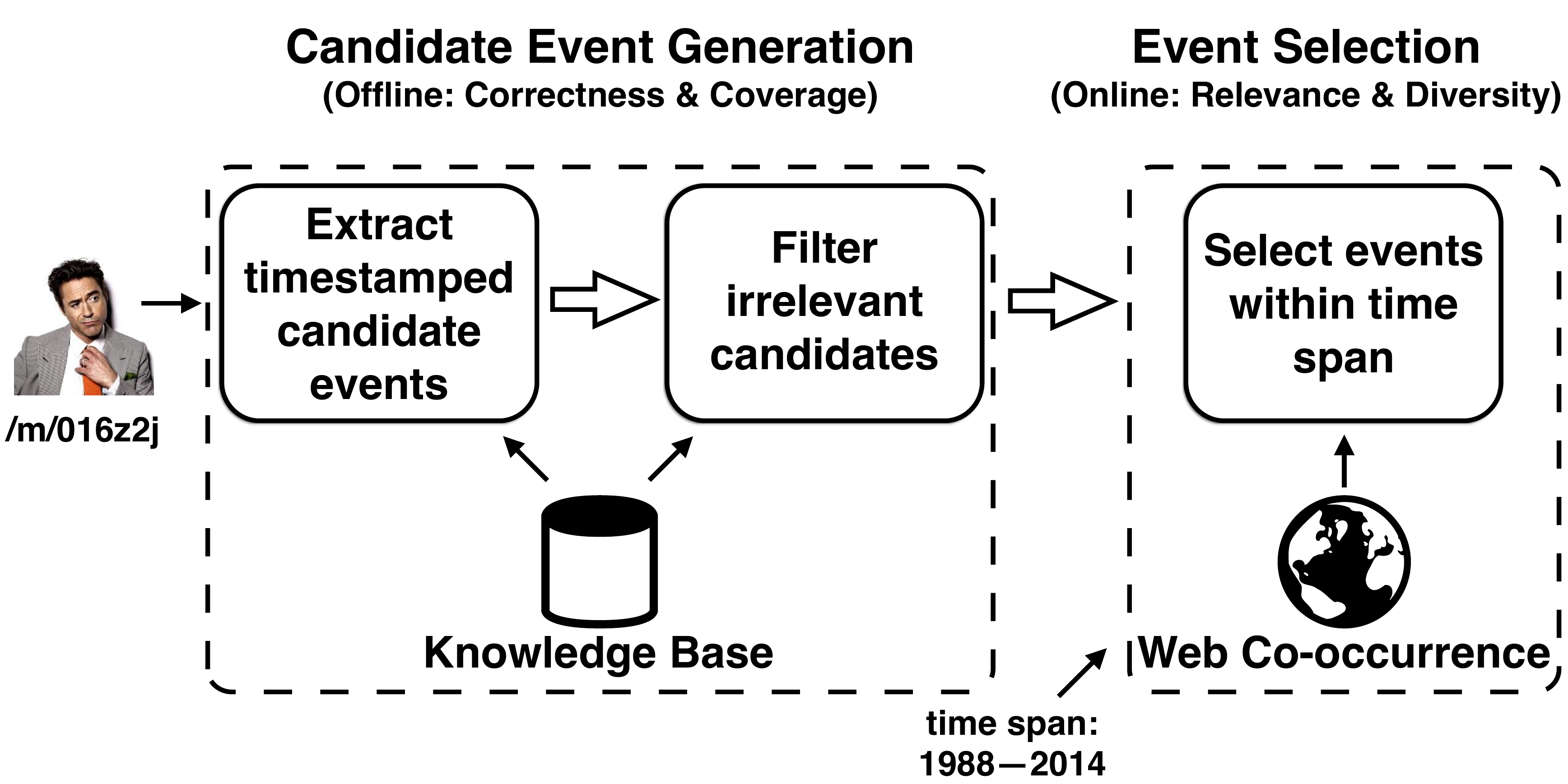}
 \vspace{-4mm}
  \caption{System architecture.
    {\sc TimeMachine} traverses the KB offline to generate candidate events for 
    a subject of interest (\eg, Robert Downey Jr.).
    At run time, the user specifies a time period of interest 
    and {\sc TimeMachine} selects a subset of events from the candidates to generate the timeline.
    }
  \label{fig:system_overview}
\end{figure}

\section{Event Candidate Generation}
\label{sec:generation}

Recall from Figure~\ref{fig:system_overview} that there are two main steps: 
candidate event generation and event selection.
In this section, we describe how we generate candidate events
given a subject of interest. 
Our approach is to generate a large set of events, and then
to filter out ``irrelevant'' ones. We give an evaluation of this
filtering step. This is a necessary preparation step for 
our key contribution in this paper (described in the following section):
dynamically selecting a subset of the remaining events at runtime,
depending on the time span of interest and the available screen real estate.


\subsection{Event Generation}
\label{subsec:event_generation}

\begin{figure}[t]
	\centering
	\includegraphics[width=0.8\linewidth]{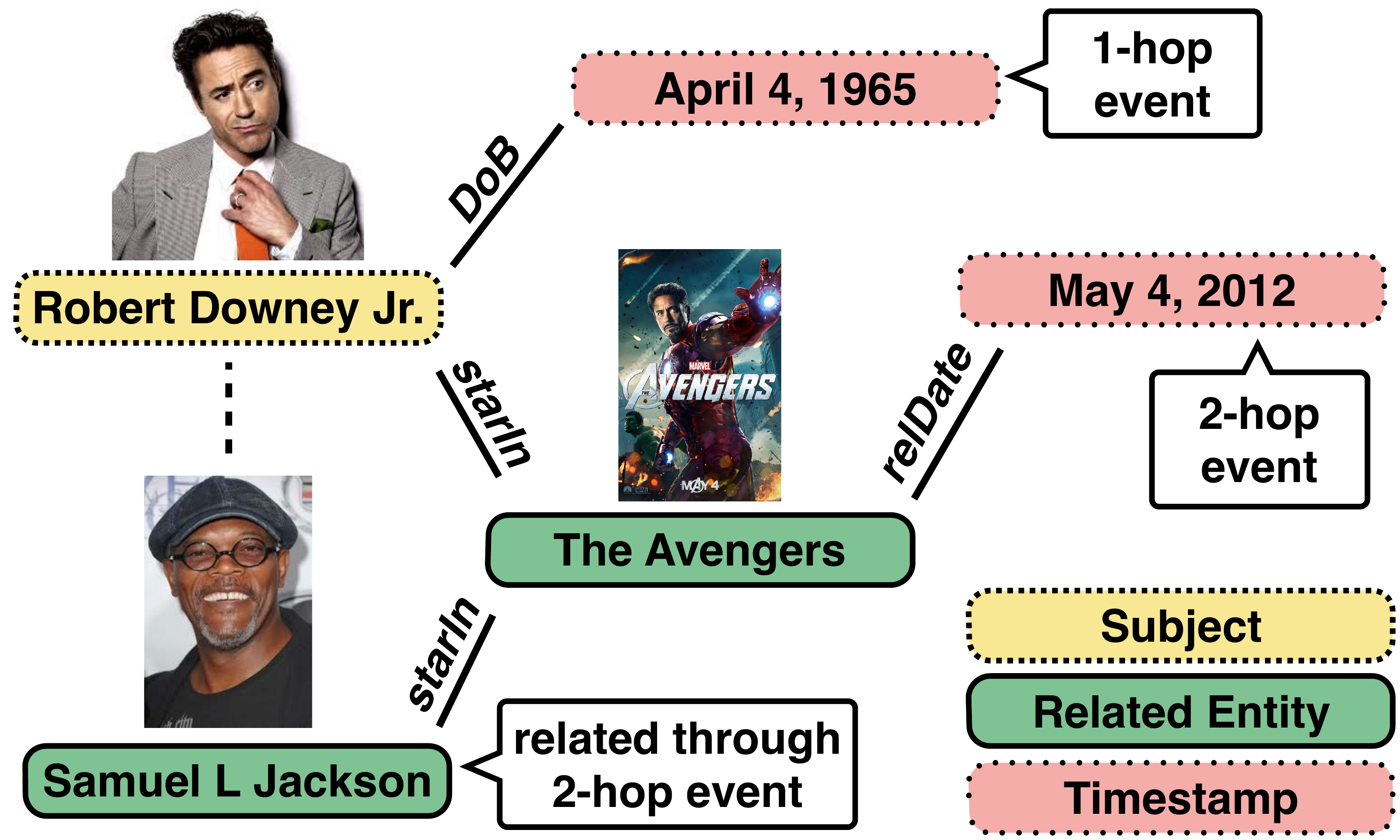}
   \vspace{-1mm}
	\caption{An illustration of the event candidate generation step.
  Events are short paths that are associated with a timestamp.
}
	\label{fig:event_generation}
\end{figure}

\begin{figure*}[t]
  \centering
    \includegraphics[width=0.8\linewidth]{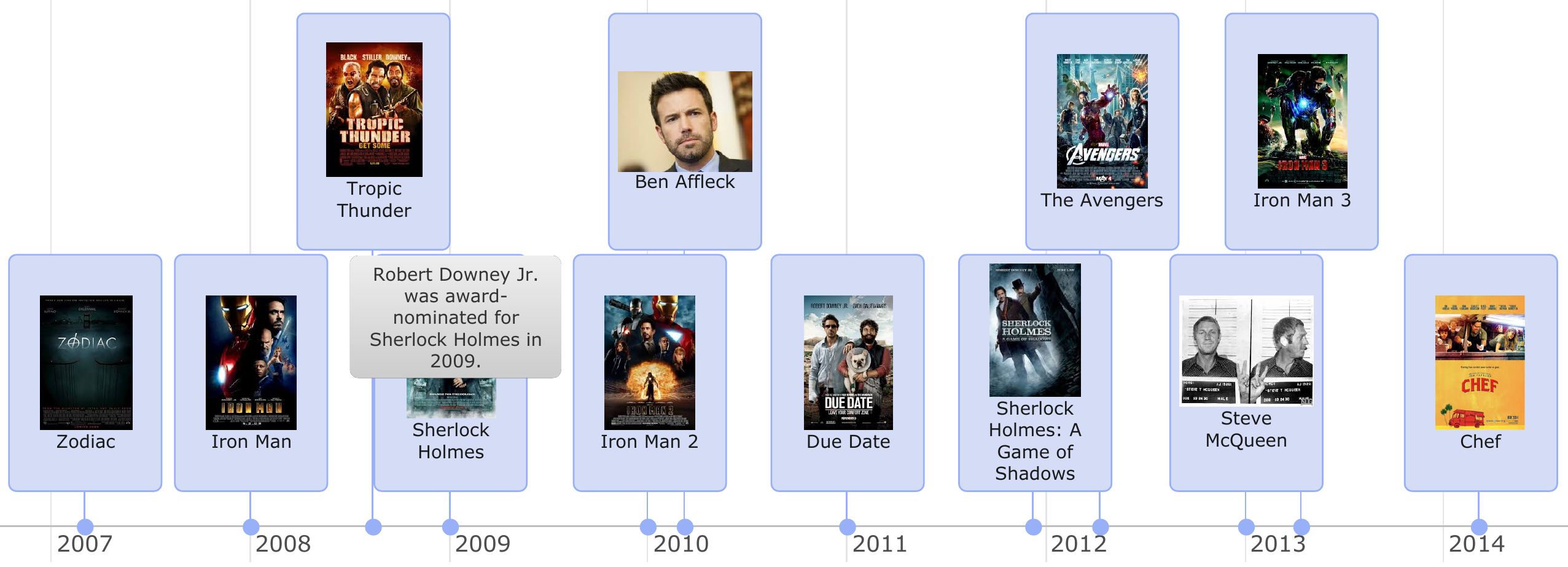}
 \vspace{-4mm}
  \caption{An example timeline for Robert Downey Jr.,
where we have zoomed in on the most recent years of his life.
Note the description for the Sherlock Holmes event (award nomination) that gets displayed on hover.
}
  \label{fig:timeline_rdjr_zoom}
\end{figure*}

We can consider the KB as a graph with nodes representing subjects and objects,
and edges representing the relationship (predicate) between the nodes.
Given a particular subject represented by a node $N_s$ in the KB, we are interested in nodes
that are connected to $N_s$ through some paths and are associated with
a timestamp; we call such paths ``events''.
As we discuss below, we consider two kinds of events: simple and
compound.
Figure~\ref{fig:event_generation} illustrates the overall process for Robert Downey Jr.

\xhdr{Simple events}
Simple events are nodes with timestamps that can be reached by paths of length one or two starting at $N_s$.
In the example in Figure~\ref{fig:event_generation}, we
can traverse along an edge of type {\em date-of-birth},
and reach a node representing the corresponding date; this is a path
of length one (1-hop event).
We can also traverse along an edge of type {\em starred-in}, and reach
the node representing the movie {\em The Avengers}. To find the corresponding date, we traverse
along a second edge of type {\em release-date} and reach a node with
the release date of the movie. This is a path of length two.

We can formally represent a simple event (derived from a path of
length two) as follows:
 $\sub
 \overset{p_1}{\longrightarrow}re_1
 \overset{p_2}{\longrightarrow} t$,
where  $\sub$ is the subject,
$re_1$ is a related entity (such as {\em The Avengers}),
$t$ is the timestamp, and $p_1$ and $p_2$ are the predicates along the path.
For simplicity, we represent 1-hop events in a similar way, by
introducing a self-loop through
 $p_1 = \mathit{self}$ and $re_1 = \sub$.
For each event $e$, we define the {\em subject} of the event as $\subFun(e) = \sub$, 
the {\em related entity} as $\reFun(e) = re_1$,
the {\em timestamp} as $\tFun(e) = t$, 
the {\em entity path} as $\pathtore{e} = p_1$,
and the {\em time path} as 
$\pathtots{e} = p_1.p_2$.



\xhdr{Compound events}
Extracting only simple events will
miss out on some implicit connections to other related entities.
For example, consider the collaboration
between Robert Downey Jr. and Samuel L Jackson in the movie {\em The Avengers} shown
in Figure~\ref{fig:event_generation}.
Even though there is no direct edge between Robert Downey Jr. and Samuel L Jackson, 
they are connected through a path as they starred in the same movie.

We can discover such connections as follows.
Suppose
we have simple events $\sub_1 \overset{p_1}{\longrightarrow}re_1 \overset{p_2}{\longrightarrow} t$ 
and
$s_2 \overset{p_3}{\longrightarrow}re_1 \overset{p_2}{\longrightarrow} t$,
which share the same related entity and timestamp,
but differ in their first hops. We join these events to generate
a new compound event
$e = \bfrac{\sub_1
  \xrightarrow{p_1}}{\sub_2
  \xrightarrow[p_3]{}} re_1
\xrightarrow{p_2} t$.
We treat $\sub_2$ as another related entity $re_2$ for $\sub_1$, and vice versa; in other words, 
$\reFun(e) = re_2 = \sub_{2}$,
and
$\pathtore{e} = p_1.p_3$.
For a discussion of implementation details please refer to the appendix of the full version of this paper~\cite{althoff2015fullpaperarxiv}.


\xhdr{Event descriptions}
To ensure an event can be understood by an end-user when they hover
over the corresponding box on the timeline, we have to
convert these paths 
into natural language form.
We do this by manually defining some templates 
for the 100 most frequently occurring paths
(see Figure~\ref{fig:timeline_rdjr_zoom} for an example).
For the remaining paths, we simply concatenate the English names of the
corresponding predicates and entities.

\vspace{-.5\baselineskip}
\subsection{Event Filtering}
\label{subsec:event_filtering}
The event generation steps we have just described may generate some irrelevant events.
For example, it can discover a path {\em ``nationality $\rightarrow$ date founded''},
so everyone with nationality {\em USA}
has a candidate event with timestamp {\em July 4, 1776},
the date on which the USA was founded.
However, arguably this event is irrelevant to most people,
since it is not specific to them, and it occurred well before many of them were born.\footnote{
%
Even for George Washington, a founding father of the USA,
it is safe to eliminate the {\em ``nationality $\rightarrow$ date founded''} path,
as there are other paths connecting him to the USA and its foundation date.
} %
We propose two simple heuristics that capture these intuitions
and filter out many irrelevant events.

The {\em Frequency Filter} uses the concept of {\em inverse document frequency}~\cite{modern-ir} 
from the IR community. The idea is that an event that is commonly associated 
with a large number of subjects is unlikely to be particularly interesting.
To formalize this, consider the set of all events $(s,re,\pi_t,t)$.
Let $N(\pi_t, re, t)$ be the number of subjects that are connected to $re$ and $t$ through path $\pi_t$, 
and let $N(\pi_t)$ be the number of distinct $(re, t)$ pairs that are connected to any subject via $\pi_t$.
Furthermore, let $C(\pi_t)=|\{(re, t) : N(\pi_t, re, t) > \theta_1 \}|$ be the number of $(re, t)$ pairs for which there are more than $\theta_1$ subjects connected through path $\pi_t$. 
Then for any given path $\pi_t$, if $C(\pi_t) / N(\pi_t) > \theta_2$, where
$\theta_2$ is some threshold, we drop that path for all subjects. 
Note that this will generalize across entities.
For example, discovering that {\em ``nationality $\rightarrow$ date founded''} is a irrelevant path 
based on people born in some countries allows us to drop
instances of this path also for people born in all other countries.


Further, entities in a KB are naturally associated with a period of \textit{existence}:
individuals are born and pass away, 
companies get founded and go out of business,
and musical groups get formed and split up.
The second filter, {\em Existence Filter}, filters out events 
that occurred before an entity began to exist.
If we find that a particular kind of path is filtered out for a large
fraction (say more than $\theta_3$) of entities,
we filter the path out for all entities. 
A canonical example is {\em ``parent $\rightarrow$ date of birth''} which obviously occurs before the subject entity is born (\ie, $\theta_3=100\%$). 
%
Based on our experiments (discussed next), we chose
$\theta_1=50$ and $\theta_2=\theta_3=0.5$, 
and we observed that slightly varying the parameters had very little impact on the results.

\vspace{-.5\baselineskip}
\subsection{Evaluation of Event Filtering}
\label{subsec:candidate_generation_evaluation_results}

\eat{
\begin{table}[tb]
\begin{center}
\setlength{\tabcolsep}{2pt}
    \begin{tabular}{lr}
    \toprule
         & \textbf{Accuracy} \\
        \midrule
        IDF Filter & 84\%\\
        Existence Filter & 100\%\\
        \midrule
        Remaining Paths & 87\%\\
        \bottomrule
    \end{tabular}
\end{center}
\vspace{-1mm}
\caption{Accuracy of two filter heuristics and of the remaining paths after filtering. Even very simple heuristics lead to good results.
\todo{Should we save the space for this table?}}
\vspace{-3mm}
\label{table:filter_accuracy}
\end{table}
}

\eat{
We will provide an extensive evaluation of our end-to-end system (including subset selection)
in Section~\ref{sec:evaluation}.
Here we just evaluate the quality of the generation and filtering steps.
}

We used Freebase~\cite{Bollacker08} to generate candidate events for four types of entities: 
music artists, actors, politicians, and athletes. We generated candidate events
for all entities of these types in Freebase and evaluated 
the quality of the results.

We evaluated the quality of our filtering using
{\em true positive / negative rate metrics} as follows.
First, for each filtering heuristic, we estimate the fraction of filtered paths that
were correctly filtered (\ie, judged irrelevant by a human) or the true negative rate.
Second, we estimate the fraction of non-filtered paths
that are correctly not filtered (\ie, judged relevant by a human)
or the true positive rate.
For each metric, we evaluated the top 100 most frequent path types 
covering over 90\% of all generated event instances 
(out of 5269 different path types generated in total).
Two domain experts manually judged each path as relevant or irrelevant.

We observe that the Frequency Filter is 84\% correct (\ie, it accidentally filters out
only 16\% of the relevant paths), and the Existence Filter is 100\% correct
(\ie, everything it filters out is irrelevant).
The main failure case for the Frequency Filter consisted of relevant events involving many entities, 
such as large award ceremonies or military conflicts.  
We also observe that among the paths that pass both filters,
87\% are correct.
The main failure case are birthdates of related people (\eg, members of the same band), 
which are arguably irrelevant to the subject.

In addition to high correctness, we need the event generation phase to have high {\em coverage}.
Figure~\ref{fig:path_coverage} 
plots the number of events on the X-axis versus the number of
entities for which we extracted at least this many events on the
Y-axis (after filtering).
Suppose we require at least 100 candidate
events for an entity before we consider it to be
``history rich'' enough for us to generate its timeline.
The figure shows that we can generate timelines for 12k entities if we
use simple events, and for 64k entities if we use compound events 
(see Section~\ref{subsec:event_generation}).
This shows that Freebase has a sufficiently rich set of events
to make our approach possible, even though it is incomplete
in many other ways \cite{Dong14}.
With the advent of systems for automated knowledge-base population
such as Knowledge Vault \cite{Dong14}, we can expect the coverage to improve further in the future.

\begin{figure}[t]
  \centering
  \includegraphics[width=0.9\linewidth]{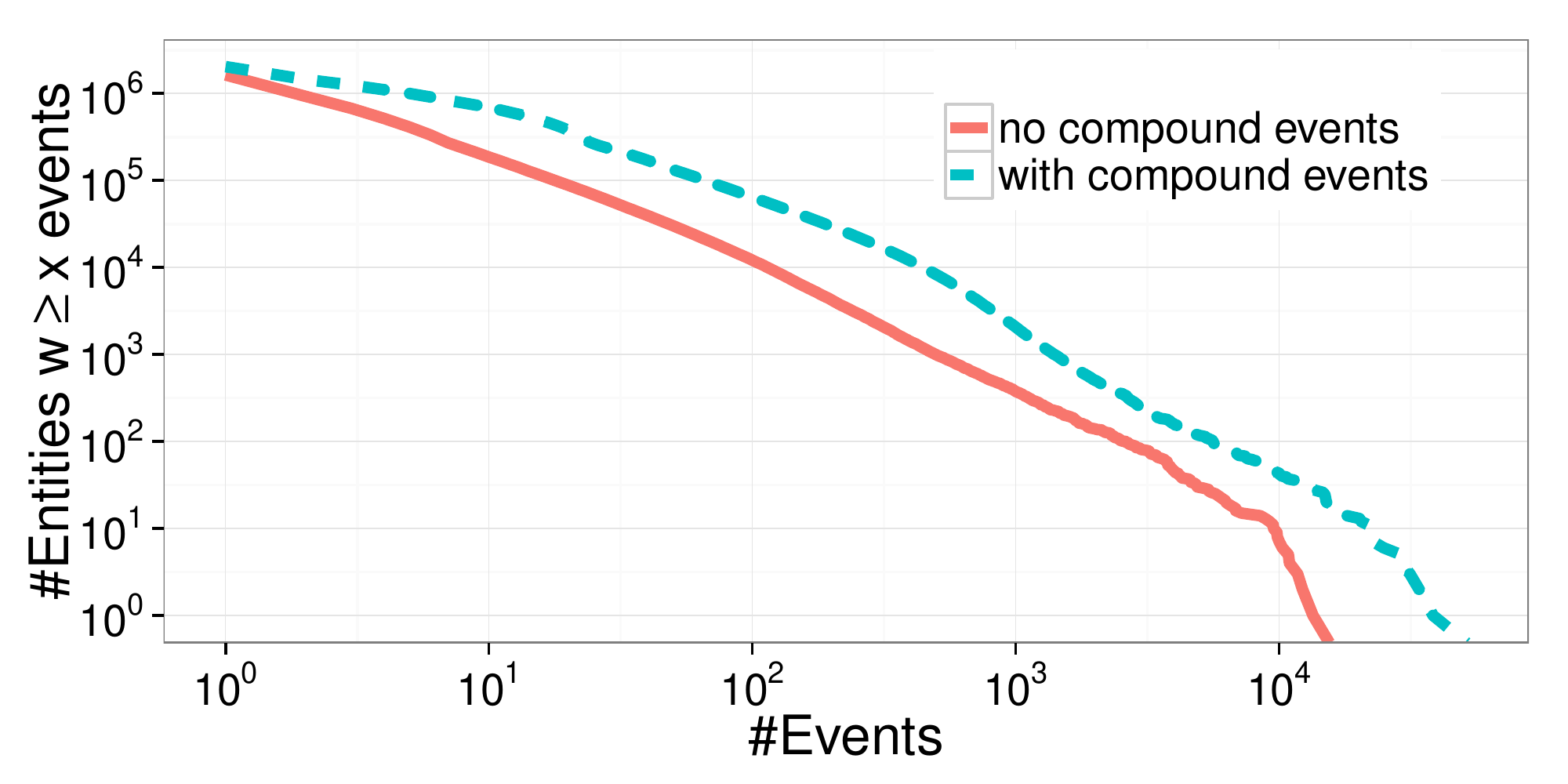} 
 \vspace{-5mm}
  \caption{Log-log coverage plot showing the number of entities with X or more candidate events 
  and illustrating the impact of adding compound events (dashed line).
   \vspace{-1mm}
  }
  \label{fig:path_coverage}
\end{figure}

\vspace{-0.75\baselineskip}
\section{Event Selection}
\label{sec:selection}

We showed in the previous section that a given entity may have hundreds of
candidate events associated with it.
In this section, we discuss our main contribution,
a way of selecting a small subset of
events to be shown on the timeline, given the
time span of interest and a specified amount of space on the screen.

Our approach will be based on submodular optimization,
which we explain in general terms in
Section~\ref{subsec:submodular_function_maximization}
(following \cite{calinescu2011maximizing,krause2012submodular}).
We define our specific optimization problem in
Section~\ref{subsec:timeline_optimization},
and give details in Section~\ref{subsec:relevance} and Section~\ref{subsec:temporal_constraint}.
In Section~\ref{subsec:optimization_algorithm},
we describe our efficient approximation algorithm,
which we prove in 
Section~\ref{subsec:approximation_guarantee} 
to yield close-to-optimal solutions.
Finally in Section~\ref{subsec:zoom},
we discuss how our algorithm enables user interactions,
such as zooming in (or out) on the timeline.

\enlargethispage{1\baselineskip} 
\subsection{Submodular Function Maximization}
\label{subsec:submodular_function_maximization}

Suppose we have a set (\eg, events) denoted by $X$,
and an evaluation function for sets $f : 2^X \rightarrow \mathds{R}$. 
Let $f_S(e) = f(S \cup \{e\}) - f(S)$ be the {\em marginal gain} of adding
element $e \in X$ to set $S$. 

A function $f : 2^X \rightarrow \mathds{R}$ is {\em submodular} if
for every pair of subsets $A \subseteq B \subseteq X$ and element $e \in X\setminus B$ we have
$f_A(e) \geq f_B(e)$.
Intuitively this means that the benefit of adding element $e$
to the smaller set $A$ is bigger than adding it to the bigger set $B$,
so $f$ exhibits the property of {\em diminishing returns}.
We restrict our attention to {\em monotone} functions;
that is, $f(A) \leq f(B)$ for all $A \subseteq B$. We also assume 
$f(\emptyset) = 0$; that is, $f$ is non-negative.

\xhdr{Constraints}
We want to compute 
$\max_{S\subseteq X} f(S)$ subject to some constraints on $S$.
A common constraint is on the size or cardinality of $S$
 \cite{krause2012submodular}.
However, in our case, we have more complex constraints,
related to temporal diversity and overlap.
To formalize these constraints, we need the notion of an
independence family, defined as follows.
An {\em independence family}
 $\mathcal{I} \subseteq 2^X$
is a family of subsets that is downward closed; that is, $A \in \mathcal{I}$ and $B \subseteq A$ implies that $B \in \mathcal{I}$. 
A set $A$ is called {\em independent} if $A \in \mathcal{I}$.
Popular independence families include {\em matroids} and
intersection of matroids \cite{calinescu2011maximizing}.
As an example, given $X = \{a,b,c\}$, an independence family is as follows: 
$\mathcal{I} = \{ \emptyset, \{a\},\allowbreak \{b\}, \{c\}, \{a,c\} \}$.


\xhdr{p-system}
For a set $Y \subseteq X$, a set $J$ is called a \textit{base} of $Y$ if $J$ is a maximal independent subset of $Y$; 
in other words $J \in \mathcal{I}$ and for each $e \in Y \setminus J$, we have $J \cup \{e\} \notin \mathcal{I}$. 
Note that $Y$ may have multiple bases, and further, a base of $Y$ may not be a base of a superset of $Y$. 
In our example of $X$ and $\cal I$, in the case of $Y=X$, the bases are $\{b\}$ and $\{a,c\}$
(since $\mathcal{I}$ does not include $\{a,b\}$, $\{b,c\}$, or $\{a,b,c\}$).

We will now use this concept to define a more general 
notion of independence families parametrized by an integer $p$,
as follows.
$(X, \mathcal{I})$ is said to be a \textit{$p$-system} if for each $Y \subseteq X$, the cardinality of the largest
base of $Y$ is at most $p$ times the cardinality of the smallest base
of $Y$:
\begin{align}
\label{eq:p_systems_property}
\frac{\max_{J:J \text{ is a base of } Y} |J|}{\min_{J:J \text{ is a base of } Y} |J|} \leq p.
\end{align}
To continue with our example of $X, Y, \cal I$, there are two bases and
$p = |\{a,c\}|/|\{b\}| = 2$.
For all other choices of $Y \subseteq X$, we have $p = 1$.
Thus, $(X, {\cal I})$ is a 2-system.
The notion of $p$-systems will be useful in Section~\ref{subsec:approximation_guarantee} to prove certain approximation guarantees.

\xhdr{Advantages}
Before continuing, it is worth discussing why we are using
submodular optimization.
A simple alternative would be to just
score each event independently, sort the events by score, and return
the top events (as long as they fit into the timeline).
Such static rankings are insufficient for our purposes since they do not consider diversity.
For instance, even though the movie {\em The Avengers} is very relevant to
Robert Downey Jr.,
the timeline should not solely consist of {\em The Avengers} events (filming, production, release date, award ceremonies, sequels, \etc).
Furthermore, we want diversity in the temporal spacing of the events
(we show in Section~\ref{sec:evaluation} that users strongly prefer diverse timelines)
which depends on the time period or zoom factor chosen by the user.

The best way to capture these effects is to reason about the entire
\textit{set} of events that should make up the timeline. 
A submodular set function allows for exactly that, and is able to
encourage different notions of diversity,
as we show in Section \ref{subsec:relevance}.
In addition, there are computational benefits to using monotone
submodular set functions.
In particular, as we show in Section \ref{subsec:approximation_guarantee}, we can devise
a greedy algorithm that enjoys certain optimality guarantees.

\subsection{Timeline Optimization Problem}
\label{subsec:timeline_optimization}

We now formalize our problem.
Let $E \subset \mathcal{E}$
be the set of all candidate events for a particular subject
entity $\sub$ constrained to events 
within the user-specified or default time span.
We will display each event as a small box of width $w$ and height $h$,
as shown in 
Figure~\ref{fig:timeline_rdjr_full}.
Assume we have available screen space of width $\mathcal{W}$ and height $\mathcal{H}$.
Let $n = \floor{\mathcal{H} / h}$ be the number of boxes that can be
stacked vertically within $\mathcal{H}$.
Our goal is to find the optimal timeline $T^*$, which we define as follows:
\begin{align}
T^* =& \argmax_{T\subseteq E} \relevance (\sub, T) \label{eq:timeline_optimization}\\
& \text{ s.t. } \temporalconstraints(T, E, \mathcal{W}, w, n) \nonumber
\end{align}
The relevance function \relevance~evaluates how relevant each event is to the
subject and how diverse they are;
this is described in more detail in Section~\ref{subsec:relevance}.
The temporal constraint \temporalconstraints~requires that all events can fit into the
provided space without overlapping or occluding each other;
this is described in more detail in Section~\ref{subsec:temporal_constraint}.

Note also that our algorithm is able to adapt to different form
factors, for example for mobile or desktop, since height and width are just
parameters to the optimization algorithm. 

\subsection{Relevance Function}
\label{subsec:relevance}
The function $\relevance(\sub,T)$
captures the quality of the selected subset of events $T$
with respect to the timeline subject $\sub$.
This is defined as a linear combination of two different kinds of
relevance functions:
\begin{align*}
\relevance (s, T) = \lambda \, \entityrelevance (s, T) 
+ (1 - \lambda) \, \daterelevance (s,T)
\end{align*}
where $0 \leq \lambda \leq 1$ trades off the importance of related
entities (\entityrelevance) versus the importance of related dates (\daterelevance).\footnote{
  We set $\lambda=0.75$ throughout,
  based on preliminary experiments on a holdout set showing that users slightly prefer
  entity relevance to date relevance.
}
We define these terms next.

\subsubsection{Entity Relevance}
\label{subsubsec:entity_relevance}

We define the relevance of a set of events $T$ to an entity $s$ as
follows:
\[
\entityrelevance (\sub,T) =
w^e_1 \, \EtoE(\sub, T)
+ w^e_2 \, \EtoEPath(T) 
+ w^e_3 \, \GtoE(T)
\]
where \EtoE~measures how relevant the specific {\em events} are,
\EtoEPath~ \linebreak[4]measures how relevant the {\em paths} are,
and \GtoE~measures how relevant the events are {\em globally} (\ie, independent of $s$).
As we show shortly, we combine \EtoE\ with \EtoEPath\ and \GtoE\ to handle data sparsity,
\cf, backoff-smoothing \cite{Katz87estimationof}.
This enables inductive reasoning that certain relationships hold
generally
when we only see a few examples of them.
For example, on average, movie roles are more relevant to actors than TV episode roles (\EtoEPath).
We discuss how we set the weight parameters in Section~\ref{sec:evaluation}.


In more detail, we measure the entity-to-entity score
 as follows:
\[
\EtoE(\sub, T) = \sum_{re \in \{\reFun(e)
  \;|\; e \in T\}}  \EtoECooc(\sub, re) 
\]
where $\EtoECooc(\sub, re)$ measures
co-occurrence of the entities, and is defined in
Section~\ref{sec:cooc}.

Since a path from a subject to a specific entity may occur too rarely
to be reliably estimated, we also consider measuring
how good the path is, by averaging the co-occurrence
score over all entities that can be
reached via all the paths in the timeline:
\begin{eqnarray*}
\lefteqn{\EtoEPath(T) } \\
 &=&
\sum_{p \in \{\pathtore{e} \;|\; e \in T\} }
\average_{e \in \mathcal{E},\, \pathtore{e} = p} \EtoECooc(\subFun(e), \reFun(e))
\end{eqnarray*}


Finally, since even the $\EtoEPath$ signal may be too sparse to
reliably estimate, we consider
\GtoE, which estimates the global importance of each entity in
the timeline:
\begin{align*}
\GtoE(T) = \sum_{re \in \{\reFun(e) \;|\; e \in T\}}  \globalimportance(re)
\end{align*}

We estimate $\globalimportance(re)$ as
the fraction of search queries that mention the entity $re$,
measured from a 3-month query log
(though other measures of global importance could be used instead).
Inferring the entity mentioned in a query is done using a proprietary
system that applies standard entity linkage algorithms
(such as \cite{Shen2015}) to the landing page of the query.

All three functions, $\EtoE$, $\EtoEPath$, and $\GtoE$, are weighted coverage functions 
defined over a {\em set} rather than a {\em multiset} of related entities (or paths to related entities).
As such, they naturally favor content diversity as duplicate entities or paths are only counted once.

\subsubsection{Date Relevance}
We define the relevance of a set of dates as follows:
\[
\daterelevance (\sub,T) = 
 w^d_1 \, \EtoD(\sub, T) 
+ w^d_2  \, \EtoDPath(T)
\]
The functions \EtoD~and \EtoDPath~ are defined in a very
similar way to their \EtoE~counterparts.
For specific dates we have
\[
\EtoD(\sub, T) = \sum_{t \in \{\tFun(e) \;|\; e \in T\}}  \EtoDCooc (\sub, t) 
\]
Recall that $\tFun(e)$ is the timestamp of event $e$, so
$\EtoDCooc$ measures how often an entity and date co-occur,
as explained below.
Then for the path level we have
\begin{eqnarray*}
\lefteqn{\EtoDPath(T) } \\
&=& \sum_{p\in \{\pathtots{e} \;|\; e \in T\}}  
\average_{e \in \mathcal{E},\,  \pathtots{e} = p}
\EtoDCooc(\subFun(e), \tFun(e))
\end{eqnarray*}

Similarly, we again use a set instead of a multiset for the timestamps and time paths
to favor temporal diversity.

\subsubsection{Web-based Co-occurrence Scores}
\label{sec:cooc}

We use co-occurrence signals between entities on the web
 to capture how related two entities are (\EtoECooc).
Similarly, we use co-occurrence between
entities and dates
(\EtoDCooc) to capture how related 
a particular date is to an entity.
We compute these quantities as follows:
\vspace{-1mm}
\begin{enumerate}
\item We run a suite of standard NLP tools (named entity recognition,
  coreference resolution, \etc) over a large corpus of 10B web
  documents using a set of in-house tools, 
  similar to the Stanford CoreNLP package.\footnote{
  \small
  \url{http://nlp.stanford.edu/software/corenlp.shtml}
  }
  We extract entity mentions (which are resolved to Freebase IDs)
and date mentions (both year and full date if available), 
using techniques similar to those described
in \cite{Shen2015}.

\item For each entity mention, we
 collect all entity-entity and entity-date co-occurrences within a
 small window around the mention (window of 100 characters or 10-12 words on average).

\item We count these co-occurrences,
and convert to probabilities (by normalizing the counts).
We define the co-occurrence scores using {\em normalized pointwise mutual
information}
\linebreak[4](\textsc{NPMI})
as follows:
\end{enumerate}
\vspace{-5mm}
\begin{align*}
\EtoECooc(\sub, re) &= \textsc{NPMI}(\sub; re) = \frac{\textsc{PMI}(\sub; re) }{- \log p(\sub, re)}\\
\textsc{PMI}(\sub; re) &= \log \frac{p(\sub, re)}{p(\sub) \, p(re)}
\end{align*}
\EtoDCooc~is defined exactly like \EtoECooc~with the
only difference that a timestamp $t$ is substituted for entity
$re$. 

PMI measures the difference between the co-occurrence probability and
the probability expected by chance if the events were independent.
It is critical to account for particularly popular entities
(\eg, Barack Obama) or dates (\eg, 2014), and dividing the co-occurrence
probability by the popularity of the co-occurring entities/dates 
is a principled way of achieving this. 

NPMI normalizes PMI to the range $[-1, 1]$.
Since we are only interested in the most related pairs of entities (or
entity/date pairs), we only retain positive \textsc{NPMI} scores.
Furthermore,
we require that this co-occurrence was extracted from at least 
five different web domains for robustness.
Computing co-occurrence statistics from a large web corpus (10B documents) and generating candidate events (for over 1M entities) takes about six hours using map-reduce.

\subsubsection{Submodularity of Objective Function}

We now show that our objective function is a monotone submodular set
function.

\begin{theorem}{}
\label{theorem:submodular_objective}
Let $f(T) = \relevance(\sub, T)$ for any given subject $\sub$.
Then $f: 2^T \rightarrow \mathds{R}^+$  is a monotone submodular set function.
\end{theorem}
\begin{proof}
First we note that $f$ is a non-negative linear combination of
weighted coverage functions 
(since $w_1^e, w_2^e, w_3^e, w_1^d, w_2^d \geq 0$).
Each of these are
known to be submodular \cite{feige1998threshold,
  krause2012submodular},
which is easy to see,
due to the diminishing returns property of weighted coverage functions.
Furthermore, a non-negative linear combination of submodular functions is submodular as well.

Second, we note that $f$ is also monotone.
This holds since all individual weighted coverage functions are
non-negative (because 
 \EtoECooc, \EtoDCooc, and \textsc{Global\-Importance}
 are all non-negative),
so adding up more terms makes the sum bigger.

Third, we have $f(\emptyset)=0$, again because all weighted coverage functions lead to empty sums.
\end{proof}



\subsection{Temporal Diversity Constraint}
\label{subsec:temporal_constraint}

The layout constraint requires that all events can fit into a timeline
of width $\mathcal{W}$ and height $\mathcal{H}$ without overlap.
We consider a simple layout strategy: 
if the boxes (of width $w$) depicting
two events have a temporal overlap,
we can stack one on top of the other,
as shown in
Figure~\ref{fig:timeline_rdjr_full},
but we require that the height of this stack be at most
$n = \floor{\mathcal{H} / h}$.

We can define this constraint more formally as follows.
Recall that each event $e \in E$ has a timestamp $\tFun(e) \in \mathds{R}$.
Let $R$ be an interval $R = [a,b] \subset \mathds{R}$.
We denote the set of events in $T \subseteq E$  with timestamps within $R$ as $T \sqcap R= \{ e\in T \;|\; \tFun(e) \in R \}$. 
We define $t_w$ as the length of a time period that corresponds to the width of $w$ on the timeline;
$t_w$ can be easily computed according to $\mathcal{W}, w$, and the beginning and ending timestamps for $E$.
Finally,
we say a set $T \subseteq E$ of events satisfies the layout constraint $\temporalconstraints(T, E, \mathcal{W}, w, n)$  if
\begin{align}
\forall t\in \mathds{R}: |T \sqcap [t, t+t_w) | \leq n \label{eq:LayoutConstraint}
\end{align}
This constraint can be interpreted as follows:
for any point in time $t\in \mathds{R}$, draw a line of width up to $t_w$ starting at $t$.
Consider the intersection of all timestamps in $T$ with this line.
If the size of the intersection is less or equal to $n$, then we know that we can vertically stack the events in the intersection without violating the height constraint.



It turns out that this constraint forms a $p$-system that enables us to prove approximation guarantees (see Section~\ref{subsec:approximation_guarantee}).
\begin{theorem}{}
\label{theorem:layout_constraints_psystem}
Let $(E,\mathcal{I})$ be an independence family based on our layout
constraint where $T \in \mathcal{I}$ if $T \subseteq E$ satisfies
Equation~\eqref{eq:LayoutConstraint}.
Then $(E,\mathcal{I})$ forms a $p$-system for $p=2$.
\end{theorem}
\vspace{-0.5\baselineskip}
\begin{proof-idea}
For the full proof please refer to the appendix of the full version of this paper~\cite{althoff2015fullpaperarxiv}.
The idea is as follows:
We need to show that $|J_{\max}| / |J_{\min}| \leq 2$ where $J_{\min}$ and $J_{\max}$ are minimal and maximal bases of an arbitrary subset $T \subseteq E$.
We can show that $J_{\max}$ can be at most twice as large as $J_{\min}$ 
by ``deleting'' all elements from $J_{\max}$ eventually,  
where in each step we delete an element $b \in J_{\min}$, 
and up to two elements in $J_{\max}$ (if they exist) in close proximity to $b$.
Intuitively this process works because we can never have an element in $J_{\max}$ that we cannot delete in this way
since then either 
$J_{\min}$ has too few points to be a maximal independent set (base),
or $J_{\max}$ has too many points to be an independent set (it would violate the layout constraint).
Both would contradict our assumption that both $J_{\min}$ and $J_{\max}$ are bases of $Y$.
\end{proof-idea}
\vspace{-2\baselineskip}

\subsection{Optimization Algorithm}
\label{subsec:optimization_algorithm}

Our problem is reduced to the problem of finding a solution $T^*$ that obtains $\max_{T \in \mathcal{I}} \relevance(s,T)$, 
where $\mathcal{I}$ is the independence family that we defined in Section~\ref{subsec:temporal_constraint} (see Theorem~\ref{theorem:layout_constraints_psystem}).
Unfortunately, such problems are NP-hard for many classes of
submodular functions,
including weighted coverage \cite{feige1998threshold} (our case).
Therefore we focus our attention
on efficient algorithms with theoretical approximation guarantees.

As we show in Section~\ref{subsec:approximation_guarantee},
a greedy algorithm (see Algorithm \ref{alg:greedy}) 
 has certain approximation guarantees.
The algorithm incrementally builds an approximate solution $\hat{T}$ (without backtracking), starting with the empty set. 
In each iteration it adds an element $e$ from the set of valid candidates $C$ that most improves the current
solution (according to the marginal gain $\relevance_{\hat{T}}(e)$),
while maintaining independence of the solution (see line 3).

\begin{algorithm}[t]
  \caption{{\sc GreedyTimeline}\label{alg:greedy}}
  \begin{algorithmic}[1]
  \State $\hat{T} \leftarrow \emptyset, C \leftarrow \emptyset$ 
    \Repeat
        \State $C \leftarrow \{e \in E\setminus \hat{T} \;|\; \hat{T} \cup \{e\} \in \mathcal{I}\}$
        \If{$C \neq \emptyset$} 
           \State $e \leftarrow \argmax_{e' \in C} \relevance_{\hat{T}}(e')$
           \State $\hat{T} \leftarrow \hat{T} \cup \{e\}$
        \EndIf
    \Until{$C = \emptyset$}
  \State \Return $\hat{T}$
  \end{algorithmic}
\end{algorithm}




This greedy algorithm has complexity $\mathcal{O}(|E|^2)$ 
(assuming computing $\relevance_{\hat{T}}(\cdot)$ takes constant time).
In practice, it can be sped up significantly in practice by using
lazy evaluations, as 
first proposed in \cite{minoux1978accelerated}
(see also \cite{leskovec2007cost} and Section 2 of \cite{krause2012submodular}).
This ``lazy greedy'' algorithm exploits the fact that the marginal gain for each element only decreases with each iteration; 
that is, $\relevance_{\hat{T}}(e) \geq \relevance_{\hat{T}'}(e)$ for $\hat{T} \subseteq \hat{T}'$, and therefore we can use previously computed values as upper bounds to save many evaluations of $\relevance_{\hat{T}}$.
We use this more efficient implementation of the greedy algorithm.

\subsection{Approximation Guarantee}
\label{subsec:approximation_guarantee}







Before we can prove our main theoretical result,
we introduce the following lemma.
It is proved in~\cite{nemhauser1978analysis} for the special
case where $\mathcal{I}$ is defined by 
the intersection of $p$ matroids on $X$,
and for the more general case of $p$-systems in Appendix B of \cite{calinescu2011maximizing}.

\vspace{-1mm}
\begin{lemma}
\label{theorem:general_approximation_ratio}
~\cite{calinescu2011maximizing,nemhauser1978analysis}
The algorithm {\sc GreedyTimeline} to compute
$\max_{S \in \mathcal{I}} f(S)$, where $(X, \mathcal{I})$ is a
$p$-system and $f : 2^X \rightarrow \mathds{R}^+$
 is a monotone submodular set function, has a tight approximation ratio of $1/(p + 1)$.
\end{lemma}
\eat{\vspace{-1.0\baselineskip}
\begin{proof}
See \cite{nemhauser1978analysis} for the special
case where $\mathcal{I}$ is defined by 
the intersection of $p$ matroids on $X$.
A proof for the more general case of $p$-systems is given in Appendix B of \cite{calinescu2011maximizing}.
These ratios are tight for all $p$ \cite{calinescu2011maximizing}.
\end{proof}
}
\vspace{-2mm}

We have shown in Theorem~\ref{theorem:submodular_objective}
that \relevance~is a monotone submodular set function  
and in Theorem~\ref{theorem:layout_constraints_psystem} that the temporal constraints 
\temporalconstraints~form a $p$-system for $p=2$.
Following Lemma~\ref{theorem:general_approximation_ratio}, our greedy algorithm
has the following approximation bound.
\begin{theorem}
\label{theorem:our_approximation_ratio}
Algorithm {\sc GreedyTimeline} has an approximation ratio of $1/3$;
that is, 
$$\relevance(s,\hat{T}) / \relevance(s,T^*) \geq 1/3,$$
for any subject $s$, where $\hat{T}$ is the output of our algorithm {\sc GreedyTimeline},
and $T^*$ is the optimal solution.
\label{theorem:our_approximation_guarantee}
\end{theorem}
\vspace{-2mm}
\eat{\vspace{-1.0\baselineskip}
\begin{proof}
We showed that \relevance~is a monotone submodular set
function  in Theorem~\ref{theorem:submodular_objective}.
The temporal constraints 
\temporalconstraints~form a $p$-system for
$p=2$ as stated in Theorem~\ref{theorem:layout_constraints_psystem}. 
Finally, invoking Lemma~\ref{theorem:general_approximation_ratio} concludes the proof.
\end{proof}
}

\subsection{Zooming in or out of the Timeline}
\label{subsec:zoom}


We have proposed an efficient algorithm for generating timelines.
This efficiency (in particular, the
``lazy greedy'' property that allows us to re-use previously computed values)
enables us to quickly (re)compute the optimal timeline
if the user chooses to dynamically zoom in or out of a specific time period. 
In practice, we observe running times roughly linear in the number of events taking a few hundred milliseconds which is much faster than the quadratic theoretical worst case bound.
An example was given in Figure~\ref{fig:timeline_rdjr_zoom},
where we show the timeline
for the most recent few years of Robert Downey Jr.'s life
(\cf, Figure~\ref{fig:timeline_rdjr_full}).

The default interval for each timeline is computed as follows:
we choose the shortest time
period that covers at least 90\% of all generated events (restricted
to the lifetime of the entity).
Note that for person entities, this time period usually corresponds to less than 90\% of their lifetime
based on the intuition that most interesting events happen to people
after they grow up, but before they retire.
(See Section~\ref{sec:future} for a discussion of when this
heuristic can fail.)



\vspace{-0.5\baselineskip}
\section{Evaluation}
\label{sec:evaluation}

In this section, we evaluate the quality of our method for producing timelines\footnote{
A demo is available at\\
\url{http://cs.stanford.edu/~althoff/timemachine}
}.
Since there is no ground truth to compare to, we asked Amazon Mechanical Turk 
raters
to vote for their preferred timeline.
We do this in a series of paired comparisons
in which we vary one component of the algorithm at a time,
resulting in six different models\footnote{
  Experiments on varying $w_1, w_2, w_3$ show that the results are insensitive to the exact parameter values as long as $w_1 \gg w_2 \gg w_3$ (\cf, backoff-smoothing \cite{Katz87estimationof}, see Section \ref{subsubsec:entity_relevance}).
},
summarized in Table~\ref{tab:experiments}.
The results are shown in 
Table~\ref{table:evaluation_results}
and Figure~\ref{fig:results_with_ci}, and are explained shortly.

In summary, our experiments show that (1) users always significantly
prefer our full method over baseline and state-of-the-art methods; 
(2) enforcing temporal diversity and content diversity significantly 
improves the results; and (3) both entity relevance and date relevance
contribute to generating a quality entity timeline.

\subsection{Experimental setup}
We generated timelines for
250 popular entities (75 music artists and bands, 75 actors, 50 politicians, 50 athletes)
for each of the six methods in Table~\ref{tab:experiments}.
We chose popular entities instead of random or tail entities because
evaluations cannot be trusted on entities that most raters are not at all familiar with.
Popular entities also account for the major share of the total query volume
and their large number of candidate events and often long lifespan makes timeline generation particularly challenging.
Furthermore, we chose to evaluate timelines through pairwise preferences rather than absolute quality judgments as this has often been found to be less subjective and thus more reliable~\cite{carterette2008here}.

Let $T(e,m)$ denote the timeline for entity $e$ generated by model $m$;
let $m=0$ denote the full model (the control), and let $m=1:5$ denote
one of the ablated models (experimental conditions; described in
Table \ref{tab:experiments} and the following sections).
For each entity $e$, we displayed the control timeline $T(e,0)$ and
the experimental timeline $T(e,m)$ for $m>0$, one above the other;
we randomized the decision whether the experiment or control was shown
on top.

We asked each rater which timeline they preferred, on a 5-point scale,
corresponding to strongly preferring the top one, slightly preferring the top one, being neutral, slightly preferring the bottom one, and strongly preferring the bottom one.
We also asked each rater to give qualitative comments to justify their
decision, to gain further insight.
Each pair of timelines is rated by five different raters (1154 distinct raters in total).
We encouraged raters to research each entity (\eg, using Wikipedia)
before evaluating each timeline about that entity; fortunately,
79\% of raters reported that they were already familiar with these entities.

To simplify the analysis, we collapsed the user
votes to a 3-point scale: prefer control (full model), neutral, or prefer experiment (ablated model).
Let $V(e,m,r) \in \{F, T, A\}$ be the vote by rater $r \in R$ for entity $e \in E$ and method $m \in M$,
where $F$ represents preferring the full model, $T$ represents a tie,
and $A$ represents preferring the ablated model.
Let $N(e,m,v) = |\{ r \in R : V(e,m,r) = v \}|$ 
be the number of raters who voted for category $v \in \{F, T, A\}$,
for entity $e$, and method $m$.
Let $M(e,m) = \argmax_v N(e,m,v)$ be the majority vote.

We compute {\em agreement} between raters (RAggr) as the fraction of raters agreeing with the majority vote
(including tie votes and tied majorities):
\[
\mathit{RAggr}(m) =
\frac{ |\{ V(e,m,r) = M(e,m) : e \in E, r \in R\} | }
{|\{V(e,m,r) : e \in E, r \in R\} | }
\]

We define the rater {\em preference} (RPref) for the full method
as the fraction of times the majority of raters vote for the full
method,
excluding cases where there is no clear majority;
that is, we set $M(e,m)= \mathrm{NULL}$ if the majority is not unique
(\eg, if we have 2 votes for F, 2 for A, and 1 for T):
\[
\mathit{RPref}(m) =
\frac{ |\{ M(e,m)  = F : e \in E \} | }
{|\{M(e,m) \in \{ F, A\} :  e \in E\} | }
\]
(Note that a 5:0 vote in favor of full (5 F, 0 A) is treated the same as a 3:2 vote.)
If both methods are equally good, 
we would expect both the full and the ablated model to win exactly 50\% of the time;
that is, $\mathit{RPref}(m)=0.5$ (our {\em null hypothesis}). 
This allows us to use a simple two-sided Binomial hypothesis
test of significance.
\begin{table}
\begin{center}
\resizebox{.9\columnwidth}{!}{
  \begin{tabular}{l|lll|ll|ll}
  \toprule
  \textbf{Name} & $\mathbf{w_1^e}$ & $\mathbf{w_2^e}$ & $\mathbf{w_3^e}$ 
  & $\mathbf{w_1^d}$ & $\mathbf{w_2^d}$ & \textbf{TD} & \textbf{CD} \\
  \midrule
  \textsc{Full} & 1 & $10^{-2}$ & $10^{-4}$ & 1 & $10^{-2}$ & 1 & 1 \\
  \textsc{Base} & 0 & 0 & 1 &0  & 0 & 1 & 1 \\
  \textsc{Full-E2D} & 1 & $10^{-2}$ & $10^{-4}$ &0  & 0 & 1 & 1 \\
  \textsc{Full-E2E} & 0 & 0 & $10^{-4}$ &1  & $10^{-2}$ & 1 & 1 \\
  \textsc{Full-TD} & 1 & $10^{-2}$ & $10^{-4}$ & 1 & $10^{-2}$ & 0 & 1 \\
  \textsc{Full-CD} & 1 & $10^{-2}$ & $10^{-4}$ & 1 & $10^{-2}$ & 1 & 0 \\
  \bottomrule
  \end{tabular}}
\end{center}
\vspace{-5mm}
\caption{Summary of experimental configurations.
We fix $\mathbf{\lambda=0.75}$ throughout.
TD = temporal diversity, CD = content diversity.
FULL-TD means the full model without temporal diversity, etc.
Note that  all methods remove duplicate events, which is a minimal
form of content diversity, but if CD=1, we ensure diversity amongst
types of events (entities and paths) as well; see Section \ref{subsec:eval_content_diversity} for details.
\label{tab:experiments}
}
\end{table}

\newcommand{\signone}{\textsuperscript{}}
\newcommand{\sigone}{\textsuperscript{*}}
\newcommand{\sigtwo}{\textsuperscript{**}}
\newcommand{\sigthree}{\textsuperscript{***}}

\begin{table}[t]
\begin{center}
  \resizebox{.9\columnwidth}{!}{
  \begin{tabular}{lrrll}
  \toprule
  \textbf{Ablated model} & \textbf{\#Tasks} & \textbf{\#Raters} & \textbf{RAggr} & \textbf{RPref} \\
  \midrule
  \textsc{Base} & 1250 & 344 & 77.0\% & 83.8\% \sigthree \\
  \textsc{Full-E2D} & 1250 &  463 & 75.7\% & 59.8\% \sigtwo \\
  \textsc{Full-E2E} & 1250 &  676 & 73.2\% & 64.3\% \sigthree \\
  \textsc{Full-TD} & 150 &  53 & 75.3\% & 86.7\% \sigthree \\
  \textsc{Full-CD} & 1250 &  665 & 81.0\% & 91.1\% \sigthree \\
  \bottomrule
  \sectionrule
  \multicolumn{4}{l}{\scriptsize{\textsuperscript{***}$p<0.001$, 
    \textsuperscript{**}$p<0.01$, 
    \textsuperscript{*}$p<0.05$}}
  \end{tabular}} 
\vspace{-2.5mm}
\caption{Summary of the user studies.
Each row shows an ablated version that was compared to the full model.
The asterisks represent the p-value corresponding to a Binomial hypothesis test
that compares the RPref value to 50\%.
\vspace{-3.5mm}
}
\label{table:evaluation_results}
\end{center}
\end{table}

\vspace{-.4\baselineskip}
\subsection{Baseline Algorithms}
\label{sec:baseline}

In our initial trial, we defined the baseline algorithm as follows:
rank all the candidate events by the \GtoE~global entity score,
and then show the top $K$ events (where each event is represented
by a box on a timeline of width 1000 pixels,
and we allow up to $n=2$ boxes to be stacked vertically).
However, we found that the \GtoE~score sometimes picked
the same related entity more than once
(\eg, if Robert Downey Jr. starred in {\em Iron Man} and later won an award for it,
we may display {\em Iron Man} twice, at two different time points).
Users strongly disliked this in preliminary experiments, so
we decided to augment the baseline through post-filtering all results by removing duplicate entities from all methods
(except Section \ref{subsec:eval_content_diversity}),
keeping the highest scoring event in each case (we never allow duplicate events).

In addition, we noticed that the baseline model would sometimes result in visual crowding,
since it does not enforce temporal diversity. 
Users strongly disliked this as well,
so we decided to further improve the baseline by post-processing all results and
removing temporally overlapping events, keeping the highest scoring
event in each case.

In practice, we can implement this modified baseline \textsc{Base} by using our
constrained submodular optimization algorithm, but setting
the weights so that $w_1^e=w_2^e=w_1^d=w_2^d=0$ and $w_3^e=1$,
thus putting all the emphasis on the \GtoE~signal. 
This is because the greedy optimization algorithm will ensure that it 
never adds an event that temporally overlaps an existing event.
Most of the time our submodular coverage function does not yield any duplicate entities.
In the rare case that the optimization does lead to duplicate entities,
we explicitly remove them to ensure that entity diversity has no impact on any other experiment 
(except Section \ref{subsec:eval_content_diversity} that explicitly measures the importance of content diversity).


We can see from the results
(Table~\ref{table:evaluation_results} and Figure~\ref{fig:results_with_ci})
that on average, 84\% of the time raters prefer our full model 
(significant difference at $p<0.001$ according to a Binomial test).
This shows that a global relevance score is inadequate, even when augmented by temporal diversity and content diversity.

\begin{figure}[t]
  \centering
  \includegraphics[width=\linewidth]{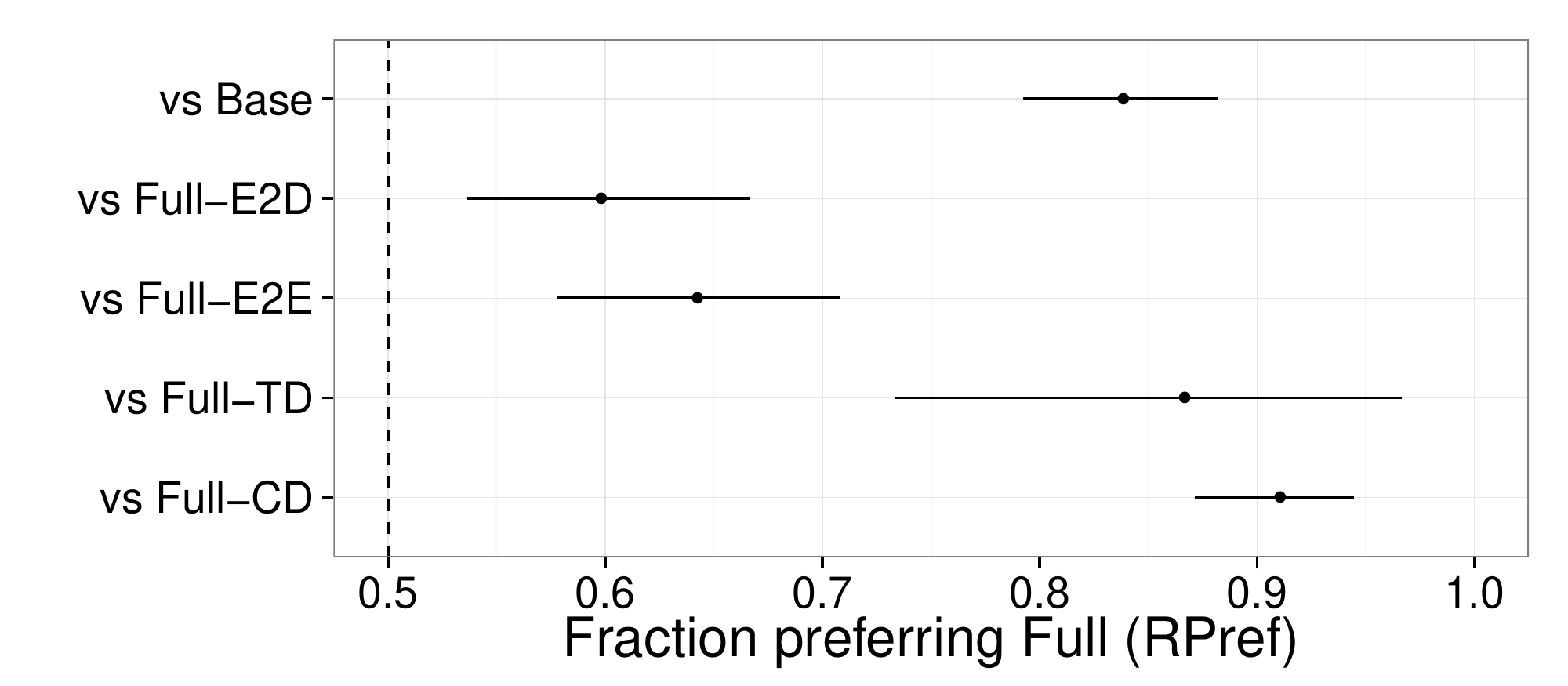}
  \vspace{-8mm}
  \caption{Fraction of entities for which raters preferred Full approach over ablated version (RPref) along with 
bootstrapped 95\% confidence intervals.
Results show significant preference for our proposed approach in all cases.
\vspace{-1mm}
}
  \label{fig:results_with_ci}
\end{figure}

The only other relevant baseline known to us is the CATE system described in~\cite{tuan2011cate}.
CATE ranks related entities by co-occurrence with the timeline entity within documents of a given context. 
This approach is very similar to the \EtoE~approach described in Section~\ref{subsubsec:entity_relevance}. 
We consider our \EtoE~signal as an improvement over the CATE baseline since, 
first, we only consider more direct connections between entities in a KB compared to co-occurrence within the same document.
Second, our relevance signal captures co-occurrence on a large web corpus within a small window which gives higher coverage and more focus compared to document-wide co-occurrence within Wikipedia only.
And third, we perform subset selection instead of a static ranking
which allows selected events to influence which other events are
selected next.
Section~\ref{subsec:eval_relevance_signals} shows that our methods
outperform \EtoE.

\subsection{Evaluating Relevance Signals}
\label{subsec:eval_relevance_signals}

To compare the  different ways of measuring event relevance,
we performed two experiments.
First, we ``turned off'' the date signal \daterelevance, by setting
$w_1^d=w_2^d=0$. We call this model Full-E2D, meaning the full model
without the E2D signal. Raters prefer our full model to this version about 60\%
of the time, which is a significant difference at the level of $p<0.01$
(Binomial test).

The utility
of the date signal depends on which kind of entity we are creating a
timeline for. For people, it is common to find the date of 
birth, death, marriage or other key events to be explicitly mentioned on the web;
this makes it relatively easy to determine that these events are
important.\footnote{Of course, our system treats birth and death dates as special, since
they inform the beginning and end of the timeline for a person (see Section \ref{subsec:zoom}).}

Second, we ``turned off'' the \entityrelevance~signal, 
by setting $w_1^e=w_2^e=0$. We call this model Full-E2E.
Raters prefer our full
model to this about 64\% of the time, indicating that the E2E signal is somewhat
more important than E2D (significant at $p<0.001$). 

However,
the benefit of the E2E signal varied by domain/vertical: we found it
most useful for actors and athletes, 
whereas for musicians, the E2D signal was more helpful.
We attribute this to conventions in what entities and dates are co-mentioned on the web (in close proximity).
E2D works well for music artists because important dates such as album release dates and tour dates are frequently mentioned across many websites (online stores, ticketing websites, etc.).
However, this is different in the movie domain.
There are many more entities related to the movie (director, producer, dozens of actors, etc.) and only a few of them will be highlighted in close proximity to the movie release date (usually one or two star actors). 
How helpful the E2D signal is depends on what usually gets mentioned in close proximity of the date,
which is subject to certain conventions and marketing decisions.
For instance, the first {\em Pirates of the Caribbean} movie (2003) has a lower E2D score for actor Johnny Depp
than later sequels even though the first movie was the bigger milestone for Johnny Depp's career.
The sequel promotions just featured Johnny Depp (who had gained in popularity) more prominently.
We found the \EtoE~signal to be more generally applicable and less influenced by such effects.

We further found that the \EtoD~signal has less utility when events do not exhibit a clear temporal focus such as long military conflicts (compared to birth/death/marriage dates or concert tours).
In these cases, the \EtoE~signal is helpful in providing additional information in the event selection phase.

\subsection{Evaluating Temporal Diversity}
\label{subsec:eval_temporal_diversity}

As we mentioned in Section~\ref{sec:baseline}, users strongly dislike
when displayed events overlap in time, since it is not easy to see the corresponding
images and descriptions. 
Indeed, we see that in 86\% of the experiments,
raters prefer our full model over 
an ablated version, which we call Full-TD,
that maximizes relevance and content
diversity but without any temporal constraint 
(significant at $p<0.001$).
This is despite the fact
that the ablated model also includes the simple overlap filter we
described in Section~\ref{sec:baseline}.
The number of events we show is controlled and set to the number of events in our Full approach,
as we aim to measure the impact of temporal diversity while controlling for the amount of information shown
(though the overlap filter may remove some of them).
The reason the full model works better is that it can take into
account the temporal overlap during the optimization process, so if
one event is removed, another (non-overlapping) event can be added instead.

\subsection{Evaluating Content Diversity}
\label{subsec:eval_content_diversity}

As we mentioned in Section~\ref{sec:baseline}, users strongly dislike
when the same entity is repeated (with different timestamps), so we
always remove such cases from all methods.
Here, we quantify the importance of content diversity in generating timelines.
Note that in addition to entity diversity there are other, 
slightly more subtle forms of content diversity that we might
wish to consider. For example, we might not want to list only the
different movies that an actor has been in, even if they all have high relevance
scores; instead we wish to include awards, TV shows, and personal relationships as well.
Our submodular set cover objective captures this by using the \EtoEPath\ 
feature, which gives higher score to a set of events with distinct path types
 (see Section~\ref{subsubsec:entity_relevance}).

To evaluate this, we consider an alternative model
in which we evaluate the score by summing over the multiset (rather
than set) of related entities (or paths to related entities),
allowing for duplicate entities or paths during the optimization
process; we call this Full-CD. We see that raters prefer our full
model 91\% of the time compared to this ablated model 
(significant at $p<0.001$).
Again, we attribute this to the fact that the full model is aware of
the penalty for duplication during the optimization process, and can
adjust its output appropriately.

\eat{
\subsection{Summary}
\label{subsec:eval_summary}
\vspace{-1\baselineskip}
\todo{
\begin{itemize}
\item We evaluate timeline quality through a series of paired comparisons 
      in which we vary one component of the algorithm at a time 
      and ask raters for relative judgments.
\item A static ranking based on a global relevance score is inadequate, 
      even when augmented by temporal diversity and content diversity.
\item The \EtoE~and the \EtoD~signal lead to significant performance improvements 
      but are outperformed by the full model representing a combination of both signals.
\item Raters show a very strong preference for temporal diversity as well as content diversity.
\end{itemize}
}
\vspace{-1\baselineskip}
}

\vspace{-0.5\baselineskip}
\section{Related Work}
\label{sec:relatedwork}


There has been much
work on extracting temporal events from text
\cite{ji2013tackling,ling2010temporal},
 and in summarizing large text corpora
such as tag streams \cite{dubinko2007visualizing}, news corpora \cite{do2012joint}, Wikipedia biographies \cite{bamman2014unsupervised}, and Wikipedia edit histories \cite{tran2014wikipevent}.
There has also been work on mining temporal patterns across such textual data sets
\cite{huet2013mining,suchanek2014semantic,weikum2011longitudinal}. 

Another body of related work concerns document summarization
\cite{allan2001temporal,sparck2007automatic}.
The evaluation of summarization approaches has always been challenging,
and measures like Rouge \cite{lin2004rouge} are often used if ground truth summaries are available.
In our case, we use paired comparisons, since we do not have ground truth.

The summarization and IR communities have identified diversity as an important quality criterion \cite{carbonell1998use,lin2012learning}.
More recently, research has focused on complementing traditional corpus-based relevance measures with signals such as social attention \cite{zhao2013timeline}.
Early work on timeline generation by Swan and Allan \cite{swan2000automatic} attempted to summarize a news corpus by displaying major events and topics along a timeline. 
In a similar spirit, Shahaf et al. \cite{shahaf2013information} have created maps of information that summarize complex storylines across news documents.
Similar techniques have been applied to scientific literature \cite{shahaf2012metro,sipos2012temporal}.
Our paper extends this line of work by using multiple relevance signals 
(based on web co-occurrence), 
as well as showing that content and temporal diversity are critical for quality timelines.

Submodular optimization has been shown to be a powerful framework for summarization \cite{kannan2014mining,lin2012learning,shahaf2012metro,shahaf2013information,sipos2012temporal},
since it naturally captures notions of diversity through its diminishing returns properties \cite{dasgupta2013summarization}.
Furthermore, there are 
efficient approximation algorithms with theoretical guarantees
\cite{ahmed2012fair,calinescu2011maximizing,krause2012submodular,leskovec2007cost,minoux1978accelerated,nemhauser1978analysis}.

Some recent work has focused on generating personalized timelines based on Facebook \cite{graus2013yourhistory} or Twitter feeds \cite{li2014timeline}.
Timelines generated based on information from KBs have been considered in \cite{mazeika2011entity,tuan2011cate,wang2010timely}; 
these papers are the ones most related to our approach.
However, there are several differences.
First, 
\cite{mazeika2011entity,wang2010timely} do not consider a ranking of individual events
(required when space is limited) nor visual space constraints, so there is no optimization algorithm involved. 
Instead, they simply display all events which is not an option in our context as each timeline entity might have hundreds of candidate events (see Figure~\ref{fig:path_coverage}).
We have empirically shown that it is absolutely necessary to address relevance, redundancy, and space constraints to generate quality timelines.
Second, \cite{tuan2011cate} considers ranking related entities but uses a different notion of relatedness (sharing many contexts rather than more direct connections in the KB). 
In this approach, it is impossible to capture relationships between selected events as the ranking is static.
To the best of our knowledge, this is the only relevant baseline and we show in Section~\ref{subsec:eval_relevance_signals} that our proposed method outperforms an improved reimplementation of this approach.
Third, none of these papers conducts any quantitative evaluation of their timelines.

Finally, we should mention that
Bing \cite{BingTimeline} has released a system called ``Timeline'' that is somewhat similar to ours.
However, there are (to the best of our knowledge) no published accounts of how their system works.
Furthermore, their timelines are static, and do not allow the user to interact with the timeline, a feature which we consider to be very important, especially for mobile browsing.

\vspace{-0.5\baselineskip}
\section{Future work}
\label{sec:future}


In this section, we suggest some directions for future work,
based in part
on the comments written by the raters.

{\bf Choosing a better default timespan.}
As we discussed in Section~\ref{subsec:zoom}, the algorithm
picks a default time span for a person that covers 90\% of their generated life events.
However, sometimes this is suboptimal.
For example, consider the US president John F Kennedy:
many important events occurred in the last few years of his life
(assassination, presidency, Cuban missile crisis, Bay of Pigs invasion, \etc). 
Our default timespan misses many of these.
In particular, his assassination, his presidency, and his marriage to
Jacqueline Kennedy Onassis are chosen first,
and these then block other important events such as Cuban Missile
Crisis or his involvement in the Vietnam war.

We address this problem by allowing the user to zoom in to the appropriate period.
Other potential solutions include 
using the E2D scores to weight some time periods higher than others, 
and including the search over suitable time periods as part of the
optimization.

{\bf Time points vs intervals.}
Our algorithm represents events based on a single point in time.
However, some events (\eg, wars) are more naturally associated with
intervals. Currently our method may pick the start or end of a war,
but might not show both, due to the diminishing returns property.
This could be fixed by modifying the algorithm to reason about
temporal intervals.

{\bf Choosing how to describe an event.}
Sometimes a related entity is connected to the subject via many different
paths, and all have the same timestamp. In this case, it is hard to
know which relation to show to the user.
For example, the system sometimes describes a date
associated with someone's death as the end of their marriage; while
technically true, this is rather unintuitive.
Another example concerns US presidents: sometimes such people are
described as being a military commander.
Again, while technically true (since the US president is also
the Chief of the Armed Forces), this is unintuitive to users.
We may be able to fix this problem by learning a ranking model
applied to particular candidate values for any given subject and
relation
or by influencing the way the data is curated.

{\bf User preferences and subjectivity.}
In some cases, raters did not agree on which timeline was best.
The reason often seems to boil down to individual preferences.
In our experiments, the biggest area of disagreement is over how much
the timeline should be focused on
professional life (\eg, jobs, albums, books) vs personal life and
relationships (\eg, marriage, children).
Users had different opinions, even for the exact same timeline subjects, which illustrates the need for personalization in this space.
One approach would be to  distinguish between
professional and personal events, and to allow some trade-off parameter
between them.

{\bf Extractive vs abstractive summarization.}
Our current approach to building a timeline is similar to ``extractive
summarization'' techniques in the NLP community, in the sense that we
select a set of events from a candidate pool.
However, sometimes this is suboptimal, since the relationship between
two entities may be more complex.
For example, Robert Downey Jr.'s father (Robert Downey Sr.) shows up on his timeline,
but is described being a co-star in a movie rather than being his
father. While technically correct, it would be more satisfying to
create an abstract summarization of the relationship, describing that 
Robert Downey Sr. is both the father and a co-star.
We leave this to future work.

\eat{
{\bf Using machine learning}. In the future, we might try to learn
the parameters of the model from data (\eg, pairwise judgments from human raters). 
(We did not do this in the  current paper due to lack of training data.)
Such machine learning approaches could also involve personalization, for instance to trade off professional career content versus personal life content individually for each user.
}

{\bf Creating timelines for collections}.
In the future, we would like to go
beyond timelines for single entities, and derive a
method to summarize collections of entities (\eg, 1930s jazz artists),
periods of time (1920s in the U.S.), or long-lasting events (World War II).



\vspace{-1.0\baselineskip}
\section{Conclusions}
\label{sec:conclusion}

We presented a system called \systemname\ for automatic timeline generation for entities in a knowledge base.
The timeline generation problem is formulated in a submodular optimization framework that jointly optimizes for relevance, content diversity and temporal diversity.
Web-based co-occurrence signals are used to determine the relevance of other entities and dates to the timeline subject.
We proved that an efficient greedy approximation algorithm achieves near-optimal performance.
The proposed approach is evaluated through a comprehensive series of user studies demonstrating that both temporal diversity and content diversity are crucial, and that web-based co-occurrence signals significantly improve over a baseline model that relies on global importance.

\xhdr{Acknowledgments}
We thank Evgeniy Gabrilovich for many helpful discussions, 
Arun Chaganty, Stefanie Jegelka, Karthik Raman, Sujith Ravi, and Ravi Kumar for their insights on submodular optimization,
Jeff Tamer and Patri Friedman for their support with the user studies,
Danila Sinopalnikov and Alexander Lyashuk for their help with the co-occurrence pipeline,
and Jure Leskovec, David Hallac, Caroline Suen, and the anonymous reviewers for their valuable feedback.


\vspace{-1.5\baselineskip}
{\small
\bibliographystyle{abbrv}
\vspace{1mm} 
\bibliography{refs}
\enlargethispage{1\baselineskip} 
}

\newpage
\normalsize
\section{Appendix}
\label{sec:appendix}

\subsection{Proof of Approximation Guarantee}
We repeat the definition from Section \ref{subsec:temporal_constraint}: A set $T \subseteq E$ of events satisfies the layout constraint $\temporalconstraints(T, E, \mathcal{W}, w, n)$  if
\begin{align}
\forall t\in \mathds{R}: |T \sqcap [t, t+t_w) | \leq n \tag{\ref{eq:LayoutConstraint}}
\end{align}

This can be checked very efficiently because of the following equivalent formulation:
Let $\tFun(Y) = \{ \tFun(y) \;|\; y \in Y \}$ be the set of timestamps for all events
in some subset of events $Y \subseteq E$.

\begin{theorem}{Equivalence of layout constraints}
\eqref{eq:LayoutConstraint} $\Leftrightarrow$~\eqref{eq:LayoutConstraintAlternative}
\begin{align}
\forall S\subseteq T, |S| \geq n+1: \max(\tFun(S)) - \min(\tFun(S)) \geq t_w \label{eq:LayoutConstraintAlternative}
\end{align}
\end{theorem}
\vspace{-3mm}
\begin{proof}
1. Show $\neg$~\eqref{eq:LayoutConstraintAlternative} $\Rightarrow$ $\neg$~\eqref{eq:LayoutConstraint}:
\begin{align*}
\exists S\subseteq T, |S|\geq n+1: \max(\tFun(S)) - \min(\tFun(S)) < t_w
\end{align*}
Let $S^*$ be a witness for such a set $S$, 
Let $t' = \min(\tFun(S^*))$ and $t'' = \max(\tFun(S^*))$ (note that $t'+t_w > t''$). 
Then with $S^* = S^* \sqcap [t', t'') = S^* \sqcap [t', t'+t_w) \subseteq T$ 
\begin{align*}
& |S^* \sqcap [t', t'+t_w) | >= n + 1\\
\Rightarrow& |T \sqcap [t', t'+t_w) | >= n + 1\\
\Rightarrow& \exists t\in \mathds{R}: |T \sqcap [t, t+t_w) | >= n + 1 \;\; \lightning
\end{align*}

2. Show $\neg$~\eqref{eq:LayoutConstraint} $\Rightarrow$ $\neg$~\eqref{eq:LayoutConstraintAlternative}
\begin{align*}
\exists t\in \mathds{R}: |T \sqcap [t, t+t_w) | >= n + 1
\end{align*} 
Let $t^*$ be one such $t$. Define $S^* = T \sqcap [t^*, t^*+t_w)$.
\begin{align*}
\exists S^*\subseteq T, |S^*|\geq n: \max(\tFun(S^*)) - \min(\tFun(S^*)) < t_w \;\; \lightning
\end{align*} 
\end{proof}

In order to show that we can get a good approximation to the optimal solution to Equation~\eqref{eq:timeline_optimization} (see Section~\ref{subsec:timeline_optimization}) through an efficient greedy algorithm 
we now show that the independence family $(E,\mathcal{I})$, where $T \in \mathcal{I}$ if $T$ satisfies Equation~\eqref{eq:LayoutConstraint} ($T \subseteq E$), is a $p$-system for $p=2$.

Let $J_{\max}$ and $J_{\min}$ be bases of $Y \subseteq E$ with $|J_{\max}| =$\\
$\max_{J:J \text{ is a base of } Y} |J|$ and $|J_{\min}| = \min_{J:J \text{ is a base of } Y} |J|$ 
(recall Section \ref{subsec:submodular_function_maximization}).
We choose $Y \subseteq E$ instead of $T \subseteq E$ in this definition to highlight that the $p$-systems property (Equation \eqref{eq:p_systems_property}) has to hold for all subsets of events and not only for the timeline subset $T$.

\begin{definition}
We say an element $a \in E$ is \textit{blocked} in independent set $T$
(below this will be either $J_{\min}$ or $J_{\max}$)
if it could not be added to $T$ without violating the layout constraint (Equation~\eqref{eq:LayoutConstraint}). 
Formally, this means $T \in \mathcal{I}$ but $T \cup \{a\} \notin \mathcal{I}$.
\end{definition}

\begin{definition}
We define a $2 \, t_w$ \textit{ball} around $e \in E$ as the following interval: $ball(e) = (\tFun(e)-t_w, \tFun(e)+t_w)$.
\end{definition}

\begin{lemma}
If $a$ is blocked in $T$ we have $(T \sqcap ball(a)) \geq n$.
\end{lemma}
\vspace{-5mm}
\begin{proof}
There must be some interval of size
$t_w$ containing $a$ that has at least $n$ elements,
otherwise $a$ would not be blocked and could be added to $T$ while
remaining an independent set. 
\end{proof}

\begin{definition}
We define the \textsc{DeletionStep} algorithm as follows:
Input: $J_{\min}, J_{\max}$\\
Output: $J_{\min}', J_{\max}'$
\begin{enumerate}
\item Let $b$ be the minimum element in $J_{\min}$. 
Let $J_{\min}' = J_{\min} \setminus \{b\}$; that is, we delete $b$.
\item Let $a_1, \dots, a_n$ be all elements in $J_{\max} \sqcap ball(b)$ in sorted order (i.e. $a_1 \leq a_2 \leq \dots \leq a_n$).
Let $J_{\max}' = J_{\max} \setminus \{a_1, a_2\}$; that is, we delete the first two elements 
(if both elements exists, otherwise just delete one or zero elements).
\end{enumerate}
\end{definition}
This algorithm is helpful in proving Theorem~\ref{theorem:layout_constraints_psystem}
for the following reason. 
If we can delete all elements from $J_{\max}$ eventually using multiple iterations of \textsc{DeletionStep} 
then $J_{\max}$ can be at most twice as large as $J_{\min}$ (which is exactly what we need to show).

\begin{definition}
We define the following \textit{invariant}:
Let $J_{\min}$ and $J_{\max}$ be minimal and maximal bases of $Y$.
\begin{align}
\forall a \in J_{\max}: J_{\min} \sqcap ball(a) \neq \emptyset \label{eq:DeletionInvariant}
\end{align}
This invariant captures the necessary condition to be able to delete every element in $J_{\max}$ eventually because for every such element we have an element in $J_{\min}$ that could still delete it. 
Next we show that this condition is actually invariant under the previously defined \textsc{DeletionStep}.
This means that we will be able to delete every element in $J_{\max}$ eventually (used in the proof of Theorem~\ref{theorem:layout_constraints_psystem}).
\end{definition}

\begin{figure}[t]
  \centering
  \includegraphics[width=0.8\linewidth]{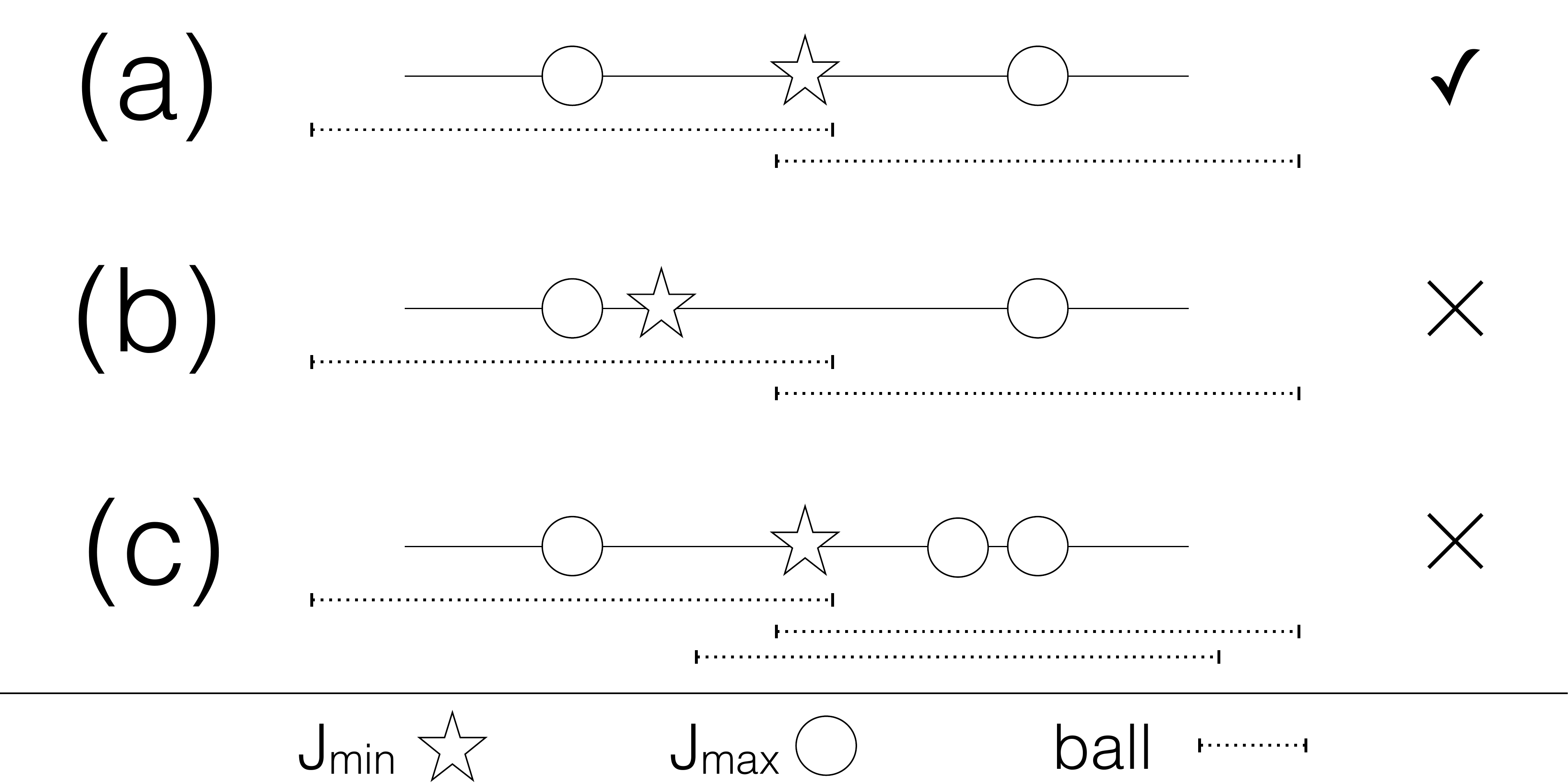}
  \vspace{-2mm}
  \caption{Intuition for proof of Lemma \ref{Lemma:DeletionInvariance} ($n=1$ for simplicity). 
  Case (a) is possible while (b) and (c) show impossible situations (see proof for details).
}
  \label{fig:proof_intuition}
\end{figure}

Intuitively, this invariant holds because we can never have an element
in $J_{\max}$ 
which is not covered in the sense of the invariant,
 since then either (1)
$J_{\min}$ has too few points to be a maximal independent set (base),
or (2) $J_{\max}$ has too many points to be an independent set.
This intuition is illustrated in Figure \ref{fig:proof_intuition} (for $n=1$).
Case (a) shows a possible situation in which each element in
$J_{\max}$ has an element in $J_{\min}$ ``in range'' and both
$J_{\min}$ and $J_{\max}$ are valid bases. 
Case (b) illustrates an impossible case in which there is no
corresponding element in $J_{\min}$ for the rightmost element of
$J_{\max}$.  
In this case the rightmost element could be added to $J_{\min}$ while remaining an independent set so $J_{\min}$ cannot be a base.
Case (c) illustrates a different impossible case in which $J_{\max}$
has too many points to be an independent set.
We have $n=1$ but there are two points in close proximity.

\begin{lemma}
\label{Lemma:DeletionInvariance}
Let $J_{\min}$ and $J_{\max}$ be minimal and maximal bases of $Y$.
Let $J_{\min}'$ and $J_{\max}'$ be the sets obtained from $J_{\min}$ and $J_{\max}$ after one or more iterations of \textsc{DeletionStep}.
The invariant (Eq.~\ref{eq:DeletionInvariant}) holds initially and after any iteration of \textsc{DeletionStep}.
\end{lemma}
\vspace{-3mm}
\begin{proof}
\textit{Invariant holds initially.\\}
Let $a \in J_{\max}$.
Initially, $J_{\min} \sqcap ball(a) \neq \emptyset$ since otherwise $a$ could be added to $J_{\min}$ ($a$ not blocked in $J_{\min}$) for a contradiction with $J_{\min}$ maximal independent set or base (cf. Figure~\ref{fig:proof_intuition} (b)). $\lightning$

\textit{Invariant holds after any number of \textsc{DeletionStep}.\\}
We now show by contradiction that the invariant holds after one or more iterations of
\textsc{DeletionStep}.
Let $J_{\min}$ and $J_{\max}$ be the bases before the \textsc{DeletionStep}(s),
and let $J'_{\min}$ and $J'_{\max}$ be the bases afterwards.
Assume the invariant holds before but not after.
Let $a^* \in J_{\max}$ be an element for which the invariant is now
violated,
so $J_{\min}' \sqcap ball(a^*) = \emptyset$.
We do a case analysis on $a^*$.


\textit{Case 1:} $a^*$ is not blocked in $J_{\min}$\\
If $a^*$ is not blocked in $J_{\min}$ we could add it to the base and obtain an independent set which is in contradiction with $J_{\min}$ being a base (cf. Figure~\ref{fig:proof_intuition} (b)); same as in the initial case above). $\lightning$

\textit{Case 2:} $a^*$ is blocked in $J_{\min}$\\
Element $a^*$ can only be blocked in
 $J_{\min}$ if $J_{\min} \sqcap ball(a^*) \geq n$.
Since we assume the invariant does not hold for the new bases,
we have  $J_{\min}' \sqcap ball(a^*) = \emptyset$,
so we must have deleted at least $n$ elements from $J_{\min}$ in
previous iterations of \textsc{DeletionStep}.
However, we now show this results in a contradiction.

We know element $a^*$ has not been deleted yet.
All of the $n$ or more elements in $J_{\min} \sqcap ball(a^*)$ would have in ``in range'' to delete $a^* \in J_{\max}$. 
Therefore, we must have deleted at least $2n$ elements from $J_{\max}$ previously.
These elements must have been in $(a^*-2t_w, a^*+2t_w)$, the range covered by the $n$ or more elements in $J_{\min} \sqcap ball(a^*)$.

Note that we delete elements in $J_{\max}$ in ascending order.
Since $a^*$ was not deleted yet, all deleted elements must be smaller than or equal to $a^*$, i.e. at least $2n$ elements from $J_{\max}$ are actually all in $(a^*-2t_w, a^*]$.
Together with $a^*$ we have at least $2n+1$ elements in $J_{\max} \sqcap (a^*-2t_w, a^*]$, an interval of length smaller than $2t_w$.
This contradicts with $J_{\max}$ being an independent set or base since our layout constraint in
Eq.~\eqref{eq:LayoutConstraint} dictates that we cannot have more
than $2n$ elements in an interval of that length
(cf. Figure~\ref{fig:proof_intuition} (c)). $\lightning$



Therefore, the invariant has to hold initially and after any number of iterations of \textsc{DeletionStep}.
\end{proof}

\begin{reptheorem}{theorem:layout_constraints_psystem}
Let $(E,\mathcal{I})$ be an independence family based on our layout constraint where $T \in \mathcal{I}$ if $T$ satisfies Equation~\eqref{eq:LayoutConstraint} ($T \subseteq E$).
$(E,\mathcal{I})$ forms a $p$-system for $p=2$.
\end{reptheorem}
\vspace{-2mm}
\begin{proof}
From Equation~\ref{eq:p_systems_property},
we need to show that $J_{\max}| / |J_{\min}| \leq 2.$

From Lemma~\ref{Lemma:DeletionInvariance} we know that the invariant holds initially.
We now recursively apply \textsc{DeletionStep}.
In each step we delete one element from $J_{\min}$ and up to two elements from  $J_{\max}$.
In the end we have $J_{\min}^\text{end} = \emptyset$ and $|J_{\max}^\text{end}| \geq |J_{\max}| - 2 |J_{\min}|$.

From Lemma~\ref{Lemma:DeletionInvariance} we know that the invariant still holds after any number of iterations.
Therefore, $J_{\min}^\text{end} = \emptyset$ implies $J_{\max}^\text{end} = \emptyset$ and we get $|J_{\max}^\text{end}| = 0 \geq |J_{\max}| - 2 |J_{\min}|$.
Because all this holds for arbitrary minimal and maximal bases $J_{\min}$ and $J_{\max}$
this is equivalent to our proposition: $|J_{\max}| / |J_{\min}| \leq 2.$
\end{proof}

This shows that we can obtain a close-to-optimal solution to Equation~\eqref{eq:timeline_optimization} using an efficient greedy algorithm (see Section~\ref{subsec:approximation_guarantee}).

\subsection{Implementation Details of the Candidate Generation Step}

The KB used in this paper, Freebase, uses Compound Value Type (CVT) nodes 
as a way to represent n-ary relations in triple form.
(In the semantic web community, CVTs are known as ``blank nodes''.\footnote{
See \url{http://en.wikipedia.org/wiki/Blank_node}.})
Such CVTs reify the n-ary relation itself as a node;
each property of this CVT node corresponds to one of the slots in the n-ary relation.
For example, the event that Robert Downey Jr. starred in {\em The Avengers} while playing the role of {\em Iron Man} is represented by the following triples:
\begin{eqnarray*}
/m/Robert  & \xrightarrow{/film/actor/film} & \texttt{CVT} \\
\texttt{CVT}  & \xrightarrow{/film/performance/character} & /m/IronMan\\
\texttt{CVT}  & \xrightarrow{/film/performance/film} & /m/Avengers
\end{eqnarray*}
This representation makes the movie more than one hop away from the actor,
and therefore makes the process outlined above a bit more complex.
To address this, we ``collapse'' the CVT nodes by replacing each path 
$a \xrightarrow{p_1} \texttt{CVT} \xrightarrow{p_2} b$ 
by a single edge 
$a \xrightarrow{p_1.p_2} b$. 

The second issue  that CVTs raise can 
be illustrated through the following example.
When a musician plays for a band, this event is represented by a CVT,
which captures the role he played (\eg, singer or drummer),
the name of the band,
and the date he joined.
If two musicians play for the same band, their corresponding CVTs will have different IDs.
This will break the above compound event heuristic.
We solve this issue by replacing the CVT ID by the ID of the corresponding band.
More generally, we replace a CVT ID by following the predicate that leads to the most diverse 
set of entities in the knowledge base (\eg, we use the band ID not the less diverse role ID).
Formally, for each ``incoming'' predicate $p_1$ we choose the ``outgoing'' predicate $p_2^*$ such that
$$
p_2^*(p_1) = \argmin_{p_2} \max_{b} |\{ b \;|\; \exists a,\texttt{CVT} : e = a \xrightarrow{p_1} \texttt{CVT} \xrightarrow{p_2} b \}|.
$$





\end{document}